%% file: main.tex
\documentclass[12pt]{asl}
\usepackage{a4wide}

\usepackage[utf8]{inputenc}
\usepackage{proof,bussproofs}
\usepackage[all]{xy}
\usepackage{graphics}
\usepackage{somedefs}
\usepackage[only,fatsemi,binampersand,bindnasrepma]{stmaryrd}
\usepackage{mathbbol}
\usepackage{xcolor}

\title[Semantical Analysis of the Logic of Bunched Implications]
{Semantical Analysis of the Logic of \\ Bunched Implications}

\keywords{Logic, Proof Theory, Model Theory, Semantics, Bunched Logic}

\author{Alexander V. Gheorghiu}
\revauthor{Gheorghiu, Alexander V.}

\address{Department of Computer Science\\University College London\\London WC1E 6BT, UK}
\email{alexander.gheorghiu.19@ucl.ac.uk}

\author{David J. Pym}
\revauthor{Pym, David J.}

\twoaddress{Department of Computer Science\\University College London\\ London WC1E 6BT, UK}{
Institute of Philosophy \\ University of London \\ London WC1E 7HU, UK}{3}
\email{d.pym@ucl.ac.uk}

\newtheorem{Theorem}{Theorem}[section]

\newtheorem{Lemma}[Theorem]{Lemma}

\theoremstyle{definition}
\newtheorem{Definition}[Theorem]{Definition}

\newtheorem{Corollary}[Theorem]{Corollary}

\newtheorem{Example}[Theorem]{Example}

\input{lbi}

\input{mk}
\newcommand{\uni}[1][V]{\mathbb{#1}}
\newcommand{\Atoms}{\mathbb{P}}

\newcommand{\set}[1]{\ensuremath{\mathbb{#1}}}
\newcommand{\system}[1]{\ensuremath{\mathsf{#1}}}
\newcommand{\logic}[1]{\ensuremath{\mathfrak{#1}}}

\newcommand{\model}[1][M]{\ensuremath{\mathfrak{#1}}}
\newcommand{\proves}[1][]{\vdash_{\system{#1}}}
\newcommand{\formulas}{\mathbb{F}}
\renewcommand{\models}[1][]{\vDash_{\mathcal{#1}}}
\newcommand{\satisfies}[1][]{\Vdash_{#1}}
\newcommand{\powerset}{\mathcal{P}}
\newcommand{\with}{\, \binampersand \, }
\newcommand{\parr}{\, \bindnasrepma \, }
\newcommand{\hash}{\#}
\renewcommand{\epsilon}{\varepsilon}
\newenvironment{scprooftree}[1]%
{\gdef\scalefactor{#1}\begin{center}\proofSkipAmount \leavevmode}%
	{\scalebox{\scalefactor}{\DisplayProof}\proofSkipAmount \end{center} }
\newcommand{\redop}[1][]{\rho_{\system{#1}}}
\renewcommand{\ll}{[\![}
\newcommand{\rr}{]\!]}
\DeclareMathSymbol{\fatcomma}{\mathrel}{bbold}{\lq\,}
\newcommand{\subbunch}{\triangleleft}
\newcommand{\subbuncheq}{\trianglelefteq}
\newcommand{\B}{\mathbb{B}}
\newcommand\wand{\mathrel{-\mkern-15mu-\mkern-9mu\ast}}
\newcommand{\munit}{\varnothing_\times}
\newcommand{\aunit}{\varnothing_+}
\newcommand{\I}{\top^{*}}
\newcommand{\weak}{\rn{w}}
\newcommand{\cont}{\rn{c}}
\newcommand{\exch}{\rn{e}}
\newcommand{\rn}[1]{\ensuremath{\mathsf{#1}}}
\newcommand*{\contadd}{\rn{c}_{\aunit}}
\newcommand*{\contmult}{\rn{c}_{\munit}}
\newcommand*{\weakadd}{\rn{w}_{\aunit}}
\newcommand*{\weakmult}{\rn{w}_{\munit}}
\newcommand*{\rrn}[1]{#1_\mathsf{R}}
\newcommand*{\lrn}[1]{#1_\mathsf{L}}

\newcommand*{\cl}[1]{\mathsf{cl}(#1)}
\newcommand*{\rcl}[1]{\cl{#1}_\mathsf{R}}
\newcommand*{\lcl}[1]{\cl{#1}_\mathsf{L}}

\newcommand{\comm}{\rn{comm}}
\newcommand{\asso}{\rn{asso}}

\newcommand{\acomm}{\comm_+}
\newcommand{\aasso}{\asso_+}
\newcommand{\mcomm}{\comm_\times}
\newcommand{\masso}{\asso_\times}

\newcommand{\cut}{\rn{cut}}
\newcommand{\taut}{\rn{taut}}
\newcommand{\ax}{\rn{id}}
\newcommand{\metatop}{\square}

\begin{document}

\begin{abstract}

This paper studies a new approach to proving soundness and completeness that bypasses truth-in-a-model to work directly with validity. Essentially, rather than working with specific worlds in specific models one reasons with eigenworlds (i.e., generic representatives of worlds) in an arbitrary model, which is handled by a calculus for validity. The method proceeds through the perspective of reductive logic (as opposed to the more traditional paradigm of deductive logic), using proof-search spaces as a medium for showing the behavioural equivalence of provability in BI's sequent calculus and validity in the given calculus. 
As BI combines intuitionistic propositional logic and multiplicative intuitionistic linear logic,  
its meta-theory is quite complex, which has resulted in much consternation with regards to its model theory; for example, the literature on BI contains many similar but ultimately different algebraic structures and satisfaction relations that either capture only fragments of the logic (albeit large ones) or have complex clauses for certain connectives (e.g., Beth's clause for disjunction rather than Kripke's). This complexity makes BI a suitable choice to demonstrate the approach to completeness. 
\end{abstract}

\maketitle

\section{Introduction} \label{sec:introduction}
 This paper centres around a new approach for proving soundness and completeness. It is an extended case-study of the method applied to \emph{the logic of Bunched Implications} (BI) \cite{OHearn1999}, which is chosen as the subject because the complexities in the logic's syntax and meta-theory help expose the more subtle aspects of how and why the method works.  Essentially, the approach proceeds by showing equivalence of provability and validity by showing \emph{behavioural} equivalence. This supports the intuition that rules for the connectives in the sequent calculus \emph{define} their meaning, since it is with these rules that the clauses of satisfaction must match.
 
As a logic, BI can be seen as arising from proof-theoretic considerations regarding the relationship between conjunction and implication, and contains primitive additive and multiplicative variants of both connectives. Consequently, contexts in BI are not lists, multisets, nor sets, they are instead \emph{bunches}, a data-structure constructed out of formulas using two context-formers, one denoting additive conjunction and one denoting multiplicative conjunction, that do not commute with each other but individually behave as expected (i.e., they are commutative and associative). The interaction between the additive and multiplicative parts of the logic renders much of the meta-theory of BI subtle and complex. 

Taking a logic to be distinguished from a language by the presence of a reasoning system, the \emph{a priori} semantics of BI in this paper is provided by the logic's sequent calculus. Sequents have two components, a context $\Gamma$ (a bunch) and an extract $\phi$ (a formula), the nomenclature being purposefully suggestive: the context is regarded as available information, and the extract as inferred information. A sequent $\Gamma : \phi $ is consequence when it has a proof in the sequent calculus $\system{LBI}$ (defined below), in which case one may write $\Gamma \proves[LBI] \phi$, the relation being called \emph{provability}. The choice of the sequent calculus over the other formalisms available for BI (e.g., Hilbert and Natural Deduction systems) is justified below.

This paper concerns the model theory of BI. One gives a model-theoretic account of BI-truth by means of a satisfaction relation of the form $w \satisfies \phi$, in which the $w$, often called \emph{possible worlds} (a terminological legacy from philosophy), are elements from a structure called a \emph{frame}, and the $\phi$ are formulas. Let $\mathcal{C}$ be a class of frames for which satisfaction is defined, then one has a relation called \emph{validity} on sequents $\Gamma : \phi $ that holds when, in any model from $\mathcal{C}$ at any world $w$, if $w \satisfies \Gamma$, then $w \satisfies \phi$. One writes $\Gamma \models \phi$ when $\Gamma : \phi $ is valid.  

This validity relation is a model-theoretic semantics for BI when it is equivalent to provability, equality being captured by \emph{soundness and completeness}: the provability relation is \emph{sound} with respect to validity when $\Gamma \proves \phi$ implies $\Gamma \models \phi$, and it is \emph{complete} when $\Gamma \models \phi$ implies $\Gamma \proves \phi$. Unpacking validity, one reads soundness and completeness as saying: the sequent $\Gamma : \phi$ is a consequence iff in the \emph{context} of just knowing that $w$ witnesses $\Gamma$ one can \emph{extract} that $w$ also witnesses $\phi$. From this perspective, one expects the proof-theoretic and model-theoretic views of logic to be reasonably close.

Various sound and complete semantics have been studied for BI in the past, including categorical, relational, and topological variants, but the most widely used models of BI in the literature employ the monoidal semantics for which completeness has been a subtle problem. Currently soundness and completeness results have been achieved only for monoids with partial or non-deterministic products or with total deterministic products, but with Beth's clause for disjunction in the definition of satisfaction. These previous results are discussed in Section \ref{sec:previous}. This paper demonstrates the completeness result for a general class of relational models with few conditions that subsumes the relational and monoidal semantics discussed above, while employing Kripke's clause for disjunction. That this is possible while avoiding the complications that arise in previously considered term-model constructions demonstrates the strength of the approach.

The intuition behind the approach to soundness and completeness in this paper is that the ways in which proof theory and model theory define the connectives coincide; for example, in both paradigms additive conjunction is defined by having that relative to some available information regarded either as context $\Gamma$ or a state of knowledge $w$ one has the conjunction $\phi \land \psi$ if and only if from the same information one has each of $\phi$ and $\psi$ independently. It proceeds by showing that whatever reasoning can be done in the proof theory can be \emph{simulated} in the model theory and \emph{vice versa}. Here, reasoning is characterized through the perspective on logic called \emph{reductive logic} --- dual to the more traditional paradigm of \emph{deductive logic} --- as explorations of \emph{proof-search spaces}. The definition of a proof-search space is contained within (see Section \ref{sec:proofsearch}), but it is only a formal treatment of the intuitive idea as the space of sequents accessible by means of reductive reasoning (i.e., by using sequent calculus rules backward) from a given sequent. The value of working with a sequent calculus for BI is that is has local correctness so that reasoning is easier to characterize.

 To formulate a proof-search space characterization of reasoning about validity, one needs a concept of a proof-object for it, which is handled by encoding it with a \emph{meta-logic} in which worlds and formulas are terms and satisfaction is a relation, an application of similar ideas used in universal algebra and applied logic (see Section \ref{subsec:encoding}). Consequently, the study of validity can bypass the concept of truth in a model because rather than working with actual worlds one can work with eigenvariables for worlds in the meta-logic, dubbed \emph{eigenworlds}, so that the reasoning being witnessed can be instantiated at any world in any model. 

The paper begins with a terse but self-contained syntactic and proof-theoretic formulation of BI given in Section \ref{sec:bi}, which also defines the BI-frame semantics; it continues with a brief summary of proof-search spaces in Section \ref{sec:proofsearch}; this is followed by an analysis of model-theoretic reasoning as captured by proof-search in a meta-logic in Section \ref{sec:modeltheoryquaclassicaltheory}; the work of the previous sections are combined in Section \ref{sec:snc} to prove soundness and completeness of BI with respect to BI-frames; the paper then recapitulates previous result on the semantics of BI in Section \ref{sec:previous}, contrasting them with the results herein; a brief review of Beth's clause for disjunction in the context of the methods of this paper is given in Section \ref{sec:beth}. The paper concludes in Section \ref{sec:conclusion} with a summary of the main theorem and thesis, and a proposal of future work.

\section{The Logic of Bunched Implications} \label{sec:bi}

In this section, we give a syntactic and proof-theoretic account of BI that provides the concept of BI-truth, as well as define the concept of a model for which we prove soundness and completeness. The first part recalls the usual sequent calculus, the second part provides the ancillary proof theory for the soundness and completeness results below, and the third introduces the models.

\subsection{Syntax}

The logic of Bunched Implications (BI) \cite{OHearn1999} can be regarded as the free combination (i.e., the fibration --- see Gabbay \cite{Gabbay1998}) of (additive) intuitionistic logic, with  connectives $\land, \lor, \to, \top, \bot$, and multiplicative intuitionistic logic, with  connectives $*, \wand, \I$. A distinguishing feature of BI is that contexts are not one of the familiar structures of lists, multisets, or sets, since the two context-formers $\fatsemi$ and $\fatcomma$ representing the two conjunctions $\land$ and $*$, respectively, do not commute with each other, though individually they behave as usual; contexts are instead \emph{bunches}  --- a term that derives from the relevance logic literature (see, for example, Read \cite{Read1988}).

	\begin{Definition}[Formulas]
		Let $\Atoms$ be a set of propositional letters. The set of formulas $\formulas$ is defined by the following grammar:
		\[
		\phi ::= \rm p \in \Atoms \mid \top \mid \bot \mid \I \mid \phi \land \psi \mid \phi \lor \psi \mid \phi \to \phi \mid \phi * \phi \mid \phi \wand \phi
		\]
	\end{Definition}
\begin{Definition}[Bunches] \label{def:bunch}
	   The set of bunches $\B$ is defined by the following:
		\[
		\Gamma ::= \phi \in \formulas \mid \aunit \mid \munit \mid \Gamma \fatsemi \Gamma \mid \Gamma \fatcomma \Gamma
		\]
		The $\fatsemi$ is the additive context-former and the $\aunit$ is the additive unit; the $\fatcomma$ is the multiplicative context-former and the $\munit$ is the multiplicative unit.
	\end{Definition}
	\begin{Definition}[BI-sequent] \label{def:sequent}
		A BI-sequent is a pair $\Gamma : \phi $ in which $\Gamma$ is a bunch, called the context, and $\phi$ is a formula, called the extract. The empty pair is also a sequent denoted $\square$.
	\end{Definition}
	
	Let $\Delta \subbunch \Gamma$ denote that $\Delta$ is a proper sub-tree of $\Gamma$, and let $\Delta \subbuncheq \Gamma$  
denote that either $\Delta \subbunch \Gamma$ or $\Delta = \Gamma$, in which case $\Delta$ is called a sub-bunch of $\Gamma$. One may 
write $\Gamma(\Delta)$ to mean that $\Delta$ is a sub-bunch of $\Gamma$. The operation $\Gamma[\Delta \mapsto \Delta']$ --- abbreviated to 
$\Gamma(\Delta')$ where no confusion arises --- is the result of replacing the occurrence of $\Delta$ by $\Delta'$. 

Since contexts are more complex than in many of the more familiar logics (e.g., classical logic, intuitionistic logic, etc), the following is an explicit characterization of the analogous structural behaviour (i.e., equivalence up-to permutation):

\begin{Definition}[Coherent Equivalence]
	Two bunches $\Gamma,\Gamma' \in \B$ are coherently equivalent when $\Gamma \equiv \Gamma'$, where $\equiv$ is the least relation satisfying: 
	
	\begin{itemize}
		\item commutative monoid equations for $\fatsemi$ with unit $\aunit$
		\item commutative monoid equations for $\fatcomma$ with unit $\munit$
		\item coherence; that is, if $\Delta\equiv\Delta' $ then $\Gamma(\Delta) \equiv \Gamma(\Delta')$. 
	\end{itemize} 
\end{Definition}
	
Bunches are typically understood as the syntax trees provided by Definition~\ref{def:bunch} modulo coherent equivalence, in the same way that lists for the contexts of classical logic sequents are understood modulo permutation.

The idea that the context-formers are representations of the conjunctions and the units are representations of the tops provides the following transformation:

\begin{Definition}[Compacting]
    The compacting function $\lfloor - \rfloor : \B \to \formulas$ is defined inductively by fixing formulas and the following action on bunches:
    \[
    \lfloor \Gamma \fatcomma \Delta \rfloor := \lfloor \Gamma \rfloor * \lfloor \Delta \rfloor \quad  \lfloor \munit \rfloor := \I \quad \lfloor \aunit \rfloor := \top  \quad  \lfloor \Gamma \fatsemi \Delta \rfloor := \lfloor \Gamma \rfloor \land \lfloor \Delta \rfloor
    \]
\end{Definition}

Since much of subsequent work concerns proof-search and the reductive view of logic, proofs are defined by a correctness criterion rather than the familiar inductive construction (see, for example, Troelstra and Schwichtenberg \cite{Troelstra}), doubtlessly already familiar. 

	\begin{Definition}[Rule]
	A rule $\rn{r}$ is a relation on sequents.
	\end{Definition}
    The situtation $\rn{r}(S,S_1,...,S_n)$ may be denoted in the following format:
	\[
	\infer[\rn{r}]{S}{S_1 & ... & S_n}
	\]
	\begin{Definition}[Sequent Calculus]
		A sequent calculus is a set of rules.
	\end{Definition}
	
	\begin{Definition}[Proof] \label{def:proof}
	 	 Let $\system{L}$ be a sequent calculus and let $S =  \Sigma : \Pi $ be a sequent. A rooted finite tree $\mathcal{D}$ of sequents is a $\textsf{L}$-proof of $S$, if for any node $\zeta$,
	 	\begin{itemize}
	 		\item if $\zeta$ is a leaf if and only if $\zeta = \square$;
	 		\item if $\zeta$ has children $P_0,...,P_n$ in $\mathcal{D}$, then there is a rule $\rn{r} \in \system{L}$ such that $\rn{r}(P_0,...,P_n,\zeta)$; and,
	 		\item if $\zeta$ is the root, then $\zeta$ is $S$.
	 	\end{itemize}
	 	The $\system{L}$-provability judgment $\Sigma \proves[L] \Pi$ holds if and only if there is an $\system{L}$-proof of the sequent $ \Sigma : \Pi $.
	 \end{Definition}

	The sequent calculus defining BI is defined as follows:
	
	\begin{Definition}[System \system{LBI}]
	    System \system{LBI} is given in Figure \ref{fig:lbi}.
	\end{Definition}
\lbi

The symbol $\square$ is used to facilitate the transition between the concept of a proof and the concept of a reduction in a proof-search space in Section \ref{sec:proofsearch}.

\subsection{Proof Theory} \label{subsec:prooftheory}
Heuristically, the provability judgment $\Gamma \proves \phi$ is an \emph{implication} that says whenever all of $\Gamma$ holds it follows that $\phi$ holds. Taking the intuition seriously, provability should be transitive and reflexive, which is captured by the admissibility of the $\cut$-rule and truth of tautologies, respectively:

\begin{Lemma}
   If $\Delta \proves[LBI] \phi$ and $\Gamma(\phi) \proves[LBI] \chi$, then $\Gamma(\Delta) \proves[LBI] \chi$.
\end{Lemma}
\begin{proof}
Proved by Brotherston \cite{Brotherston2012} (see also Gheorghiu and Marin~\cite{Gheorghiu2021}).
\end{proof}
\begin{Lemma}\label{lem:taut}
 For any $\Gamma \in \B$ the following holds: $\Gamma \proves[LBI] \lfloor \Gamma \rfloor$.
\end{Lemma}
\begin{proof}
    This follows from induction on the size of $\Gamma$ --- see, for example, Gheorghiu and Marin~\cite{Gheorghiu2021}.
\end{proof}

The method for proving soundness and completeness of BI with respect to a general class of models in this paper relies on showing that the proof theory and the model theory behave in the same way, but that does not mean that $\system{LBI}$ behaves as the clauses defining satisfaction in Section \ref{subsec:modeltheory} below. The underlying ideology that the rules of the sequent calculus \emph{define} the connectives (and context-formers) of the logic. From this perspective, the exchange rule $\exch$ is not tractable since its definition outsources the key behaviour of the syntax that it concerns. Hence, it is reformulated to be more suitable:

\exchrules

\begin{Lemma}\label{lem:lbiprime}
 The rules in Figure \ref{fig:exch} are admissible for BI. 
\end{Lemma}
\begin{proof}
Since $\Gamma \equiv \Gamma'$ if and only if there is a sequence of steps using the commutative monoid axioms and all of these are encoded by the new rules of Figure \ref{fig:exch}, these rules are admissible.
\end{proof}

Below are some generalizations of rules to facilitate subsequent discussion:

\begin{Lemma}
   The following rules are admissible for BI:
   \[
   \infer[\taut]{\Gamma \fatsemi \phi \proves \phi}{} \qquad \infer[\weak \rrn*]{\Gamma \fatsemi (\Delta \fatcomma \Delta') \proves \phi * \psi}{\Delta \proves \phi & \Delta \proves \psi} \qquad  \infer[\lrn \bot]{\Gamma(\bot) \proves \phi}{\Gamma(\phi) \proves\phi}
   \]
\end{Lemma}

 We may refer to $\weak \rrn *$ simply by $\rrn *$. 
 
\begin{proof}
The first two rules are admissible by combining $\weak$ with Lemma \ref{lem:taut} and $\rrn*$, respectively. The remaining rule is demonstrated to be admissible by the following derivation:

\begin{prooftree}
\AxiomC{$\square$}
\RightLabel{$\lrn \bot$}
\UnaryInfC{$\bot \proves \phi$}
\AxiomC{$\Gamma(\phi) \proves \chi$}
\RightLabel{$\cut$}
\BinaryInfC{$\Gamma(\bot) \proves \chi$}
\end{prooftree}

\end{proof}

The technical results in this section witness the following:

\begin{Definition}[System $\system{sLBI}$]
    System $\system{sLBI}$ is composed of the rules in Figure \ref{fig:slbi} in which $\asso$ is invertable.
\end{Definition}
\begin{figure}[t]
    \fbox{
    \scalebox{0.9}{
	\begin{minipage}{\textwidth}
		\centering 
		\vspace{0.2cm}
	    \[
	    \begin{array}{cc}
		 \infer[\lrn \land]{  \Gamma(\phi\land\psi)  :  \chi}{ \Gamma(\phi\fatsemi\psi) :  \chi}
		  &
		\infer[\rrn \land]{ \Gamma :  \phi \land \psi}{ \Gamma :  \phi  &  \Gamma :  \psi}
		\\[1.5ex]
		 \infer[\lrn \ast]{  \Gamma(\phi*\psi) :  \chi }{  \Gamma(\phi\fatcomma\psi) :  \chi}
		  &
		     \infer[\rrn \ast]{ \Gamma \fatsemi (\Delta_1 \fatcomma \Delta_2) :  \phi_1 * \phi_2}{ \Delta_1 :  \phi_1  &  \Delta_2 :  \phi_2}
		  \\[1.5ex]
		  \infer[\lrn \lor]{  \Gamma(\phi\lor\psi) :  \chi}{ \Gamma(\phi) :  \chi &  \Gamma(\psi) :  \chi}
		   &
		 	\infer[\rrn \lor]{ \Gamma :  \phi_1 \lor \phi_2}{ \Gamma :  \phi_i}
		\\[1.5ex]
		 \infer[\lrn \to]{ \Delta \fatsemi \phi \to \psi :   \chi}{ \Delta :  \phi &  \Gamma(\Delta,\psi) :  \chi }
		  &
		  \infer[\rrn \to]{ \Gamma :  \phi \to \psi}{ \Gamma \fatsemi \phi:  \psi}
		 \\[1.5ex]
\infer[\lrn \wand]{  \Gamma(\Delta_1 \fatcomma \Delta_2, \phi\wand\psi) :  \chi }{  \Delta_2 :  \phi &  \Gamma(\Delta_1 \fatcomma \psi) :  \chi}
 &
       \infer[\rrn \wand]{ \Gamma :  \phi \wand \psi}{ \Gamma \fatcomma \phi :  \psi}
       \\[1.5ex]
		 \infer[\lrn \I]{ \Gamma(\Delta \fatcomma \munit) :   \chi }{ \Gamma(\Delta) :   \chi}
                 &
                 \infer[\lrn \I]{ \Gamma(\Delta) : \chi }{  \Gamma(\Delta \fatcomma \munit) :  \chi}
         \\[1.5ex]
		 \infer[\lrn \top]{ \Gamma(\Delta) :  \chi }{ \Gamma(\Delta\fatsemi \aunit) :  \chi}
		  &
		       \infer[\rrn \top]{ \Gamma :  \top}{\square}
    \\[1.5ex]
                \infer[\lrn \bot]{ \Gamma(\bot) :   \chi }{  \Gamma(\phi) :   \chi }   
                 &
                 \infer[\ax]{ \Gamma\fatsemi \phi :   \phi}{ \square }
       \\[1.5ex]
        \infer[ \comm]{ \Gamma(\Delta_1 \fatcomma \Delta_2) :  \chi}{  \Gamma(\Delta_2 \fatcomma \Delta_1) :  \chi }
         & 
        \infer[\asso]{ \Gamma(\Delta_1 \fatcomma (\Delta_2 \fatcomma \Delta_3)) :  \chi}{  \Gamma((\Delta_1 \fatcomma (\Delta_2) \fatcomma \Delta_3) :  \chi }
        \\[1.5ex]
        \infer[\exch_\rn{1}]{ \Gamma(\Delta_1 \fatsemi \Delta_2) :  \chi}{  \Gamma(\Delta_2 \fatsemi \Delta_1) :  \chi }
         & 
        \infer[\exch_\rn{2}]{ \Gamma(\Delta_1 \fatsemi (\Delta_2 \fatsemi \Delta_3)) :  \chi}{  \Gamma((\Delta_1 \fatsemi \Delta_2) \fatsemi \Delta_3) :  \chi }
        \\[1.5ex]
        \infer[\cont]{ \Gamma(\Delta) : \chi }{  \Gamma(\Delta \fatsemi \Delta) :  \chi}
         &
        \infer[\weak]{ \Gamma(\Delta \fatsemi \Sigma) :  \chi}{  \Gamma(\Delta) :  \chi }
        \end{array}
        \]
        \\[1.5ex]
        $
        \infer[\cut]{ \Gamma(\Delta) : \chi }{  \Delta :  \phi &  \Gamma(\phi) :  \chi}
        $
	    \vspace{0.2cm}
    \end{minipage}
    }
    }
    \caption{System \system{sLBI}}
    \label{fig:slbi}
\end{figure}

\begin{Lemma}\label{lem:slbi}
   $\Gamma \proves[LBI] \phi$ iff $\Gamma \proves[sLBI] \phi$.
\end{Lemma}

\subsection{Model Theory} \label{subsec:modeltheory}

In model theory one thinks of statements (i.e., formulas of the logic) not as being universally true, in contrast to classical validity, but instead true with respect to a certain state of affairs such as at a certain time or with respect to some information. For example, intuitionistic logic (IL) is the constructive fragment of classical logic wherein a formula is true when one can provide a method for witnessing it; hence, the states in IL's model-theoretic semantics may be thought of as witnesses of these methods 
and the clauses of the satisfaction relation specify how the witnesses 
relate to each other --- see Dummett 
\cite{Dummett2000} for details. 

The model-theoretic semantics of BI extends the models of IL by allowing witnesses to be decomposed. A witness $w$ satisfies an additive conjunction $\phi \land \psi$ when it satisfies both $\phi$ and $\psi$; a witness $w$ satisfies a multiplicative conjunction $\phi * \psi$ when there are two states $u$ and $v$ in relation $R$ to $w$ such that $u$ satisfies $\phi$ and $v$ satisfies $\psi$. Intuitively, therefore, its semantics should be an extension of the semantics of IL that has precisely enough structure to provide a suitable denotation for the new conjunction. This handling of the semantics by means of a relation is in the style of Routley and Meyer \cite{routley1973semantics} and Urquhart \cite{urquhart1972semantics} for relavant logics.

\begin{Definition}[BI-frame]
    A quintuple $\model := \langle \uni, e, \pi, \prec, R\rangle$ is a BI-frame when $\uni$ is a set, $e$ and $\pi$ are distinguished element of the set,  $\preceq$ is a preorder on the set dominated by $\pi$ (i.e., for any $w$ in the set, $w \preceq \pi$), and $R$ is a ternary relation on the set, satisfying the following conditions:
    \begin{itemize}
        \item (\emph{Unitality}) $R(w,w,e)$
        \item (\emph{Commutativity}) $R(x,y,z)$ iff $R(x,z,y)$
        \item (\emph{Associativity}) if  $R(x,w,y)$ and $R(y,u,v)$, then there exists a $z$ such that $R(x,z,v)$ \text{ and } $R(z,w,u)$
    \end{itemize}
\end{Definition}

Following the same intuition as above, the clauses of the satisfaction relation for BI should be those for IL together with additional clauses that hand the new connectives using the available additional structure.

\begin{Definition}[Satisfaction] \label{def:satisfaction}
    Let  $\model := \langle \uni, e, \pi, \prec, R\rangle$ be a BI-frame. Given a mapping $\ll - \rr: \uni \to \powerset (\Atoms)$, called an interpretation, satisfaction in a  $\model$ is a binary relation $\satisfies$ between the worlds $\uni$ of the frame and the formulas $\formulas$ of BI defined by the clauses in Figure \ref{fig:sat:BI}. 
\end{Definition}

\begin{figure}[t]
    \centering
    \fbox{
    \begin{minipage}{0.95\linewidth}
    \[
    \begin{array}{lcc}
        w \satisfies \rm p  & \text{ iff } & \rm p \in \ll w \rr \\
        w \satisfies \phi \land \psi & \text{ iff } & w \satisfies \phi \text{ and } w \satisfies \psi \\
        w \satisfies  \phi \lor \psi & \text{ iff } & w \satisfies \phi  \text{ or } w \satisfies \psi \\
        w \satisfies \phi \to \psi & \text{ iff } & \text{ for any } v,  \text{ if } w \preceq v  \text{ and } v \satisfies \phi,  \text{ then }  v \satisfies \psi \\
        w \satisfies \top & \text{ iff } & w \in \uni \\
        w \satisfies \I & \text{ iff } & e \preceq w   \\
        w \satisfies \bot & \text{ iff } & w = \pi \\
        w \satisfies \phi * \psi & \text{ iff } & \text{ there are } u, v \text{ st.} \, R(w,u,v) \text{ and } u \satisfies \phi \text{ and } v \satisfies \psi \\
        w \satisfies \phi \wand \psi & \text{ iff } & \text{ for any } u, v, \text{ if } R(v,w,u)  \text{ and }  u \satisfies \phi, \text{ then }  v \satisfies \psi \\
    \end{array}
    \]
    \end{minipage}
    }
    \caption{Satisfaction for BI}
    \label{fig:sat:BI}
\end{figure}

For a BI-frame to become a model of BI, its components should behave as what they interpret. For example, the preorder ought to be (generally) persistent:
\[
\text{for any $\phi \in \formulas$ and any $w,u \in \uni$, if $w \preceq u$ and $w \satisfies \phi$, then $u \satisfies \phi$}
\]
And $\pi$ ought to be absurd:
\[
\text{for any $\phi \in \formulas$, if there is a world $w$ such that $w \satisfies \phi$, then $\pi \satisfies \phi$}
\]

\begin{Definition}[Model] \label{def:bialg}
    A pair $\langle \model, \ll - \rr \rangle$ in which $\model := \langle \uni, \preceq, R\rangle$ is a BI-frame and $\ll - \rr: \uni \to \powerset (\Atoms)$ is an interpretation is a model when it is persistent and, for any $\phi \in \formulas$, $\pi \satisfies \phi$.    The set of all models is $\mathcal{C}$. 
\end{Definition}

The concept of a model given in Definition \ref{def:bialg} actually arises from the approach to completeness that this paper demonstrates in that the clauses are designed to reflect the proof-theoretic behaviour of the connectives (see Section \ref{sec:snc}). This definition of a frame for modelling BI based on a relation $R$ is more general than that studied by Galmiche et al. \cite{Galmiche2005}, whose relationship to the present structure is discussed in Section \ref{sec:previous}. Similar models to the present one have previously been studied by Docherty and Pym \cite{Docherty2018,Docherty2019a, Docherty2019}. In that work, certain variations of satisfaction are also considered that may also be understood from the approach to completeness in this paper, but they are more complex without being more informative for our purposes. 

In terms of constructing models, the most difficult requirement to satisfy is persistence, but the aforementioned authors have also given conditions under which this condition can be met. A necessary condition is bifunctoriality: 
\[
\text{if $u \preceq u'$, $v \preceq v'$, $R(w,u,v)$ and $R(w',u',v')$, then $w \preceq w'$}
\]
These concerns are discussed further in Section \ref{subsec:simon}

The satisfaction relation on models defines a semantics as follows:

\begin{Definition}[Validity] \label{def:semanticrelation}
    Two formulas, $\Gamma$ and $\phi$, are in semantic relation, denoted $\Gamma \models \phi$, if, for any model $\model \in \mathcal{C}$, at any world $w$, if $w \satisfies \Gamma$, then $w \satisfies \phi$.
\end{Definition}

The equivalence of $\models$ with the the provability relation $\proves$ is the main theorem of the paper. The earlier claim in Section \ref{subsec:prooftheory} that $\proves$ behaves as classical implication is witnessed by the implication in the definition of the semantics.

The intuition for why the class of frames just defined should be complete is that it has precisely the structure required to simulate BI in the context-extract reading. In particular, the strength of universal condition on frames and on worlds in the Definition \ref{def:semanticrelation} is curtailed by the structure of the implication that assumes $w \satisfies \Gamma$, because the fact that $w \satisfies \phi$ follows has not so much then do with what is known about the frame as what is known about $\Gamma$, but this is precisely the context-extra reading in the proof-theoretic setting.

\section{Reductive Logic \& Proof-search} \label{sec:proofsearch}
The traditional paradigm of logic proceeds by inferring a conclusion from 
established premisses using an \emph{inference rule}. This is the paradigm known as \emph{deductive logic}: 
\[
\infer[\Downarrow]{\text{Conclusion}}{\text{Established Premiss}_1 & ... & 
\text{Established Premiss}_n }
\]

In contrast, the experience of the use of logic is often dual to deductive logic in the sense that it proceeds from a putative conclusion to a collection of premisses that are 
sufficient for that conclusion to be deduced using a \emph{reduction operator}. 
This is the paradigm known as \emph{reductive logic}: 
\[
\infer[\Uparrow]{\text{Putative Conclusion}}{\text{Sufficient Premiss}_1 & ... & \text{Sufficent Premiss}_n }
\]
Reductions may correspond to inference rules, read from conclusion to premisses, or 
may have other forms --- see, for example, Milner's theory of tactical reasoning \cite{Milner1984}. The process of constructing a proof in reductive logic is \emph{proof-search}. 

Historically, the deductive paradigm has dominated since it exactly captures the meaning of truth relative to some set of axioms and inference rules, and therefore is the natural point of view when considering \emph{foundations of mathematics}, the original \emph{raison d'\^{e}tre} of logic. Of course, it is the reductive paradigm from which much of computational logic derives, including various instances of automated reasoning. The value of reductive logic as a paradigm in which to study meta-theory is the moral of this paper.

There are many ways of studying (reductive) proof-search, and a number of models have been considered, especially in the case of classical and intuitionistic logic (see, for example, Pym and Ritter \cite{Pym2004}). A generic approach to representing and understanding the structure of the proof-search space is the use of co-inductive derivations trees, which have their origin in Kowalski's study of logic programming \cite{kowalski1979book}. A (co-)algebraic treatment has been considered by the authors previously \cite{Samsonschrift}, generalizing earlier work by Komandantskaya et al.~\cite{Komendantskaya2010}.

Given a sequent calculus $\system{L}$ and a space of sequent $\mathbb{S}$, define the reduction operator $\redop[L]: \mathbb{S} \to \powerset \powerset \mathbb{S} $ as follows:
\[
\redop[L]: S \mapsto \bigcup_{\rn{r} \in \system{L}}\{ \{S_1,...,S_n\} \mid \rn{r}(S,S_1,...,S_n) \}
\]

\begin{Definition}[Proof-search Space]
Let $\system{L}$ be a sequent calculus. The proof-search space of a sequent $S$ is the tree corecursively generated as follows: 

\begin{itemize}
\item The root of the tree is $S$; 
\item each element of $\bullet \in \redop[L](S)$ is a child of $S$;
\item each element in the $S_i \in \bullet \in \redop[L](S)$ is a child of the $\bullet$;
\item each node $S_i \in \bullet \in \redop[L](S)$ has a proof-search-tree of $S_i$ as a child.
\end{itemize}

\end{Definition}

\begin{Example} \label{ex:pss}
Below is a section of the proof-search space for $ \rm p \fatsemi \rm \aunit \fatsemi \rm p \to \rm q :  \rm q $ --- the search proceeds in the direction of the arrows. The $\bullet$-nodes represents a set of sufficient premises for a particular instance of the rule(s) labelled on the arrow and $\rn{struct.}$ is used as a shorthand for the various structural rules that may apply (e.g., rules from $\cont, \weak, \contadd, \contmult$, and $\acomm$):
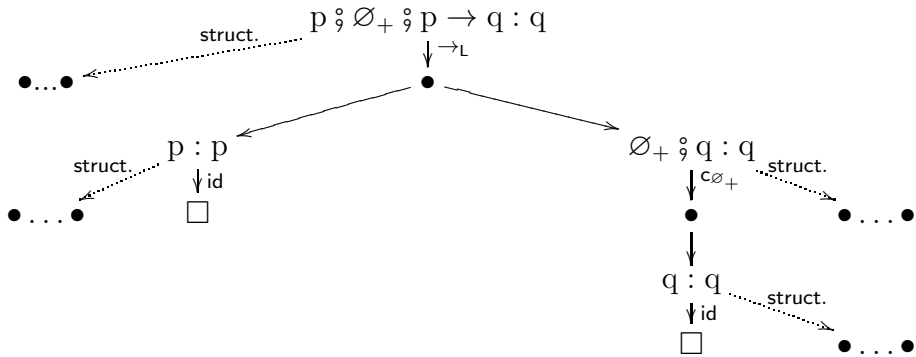
\[
\xymatrix@R=0.35cm{
 & & {  \rm p \fatsemi \aunit \fatsemi \rm p \to \rm q:\rm q  } \ar[d]^{\lrn \to} \ar@{..>}[dll]_{\rn{struct.}}  & & \\
 \bullet ... \bullet  & &   \bullet \ar[dr] \ar[dl]  & & \\
 & \ar@{..>}[dl]_{\rn{struct.}}   \rm p: \rm p  \ar[d]^{\ax} & & \ar@{..>}[dr]^{\rn{struct.}}  \aunit \fatsemi \rm q:\rm q  \ar[d]^{\contadd} & \\
\bullet \hdots \bullet  & \square & & \bullet \ar[d] & \bullet \hdots \bullet  \\
  &  & &  \rm q : \rm q  \ar[d]^{\ax} \ar@{..>}[dr]^{\rn{struct.}} &   \\ 
&  & & \square &  \bullet \hdots \bullet \\
}
\]
\end{Example}

The $\bullet$-nodes are suggestively dubbed \emph{or}-nodes and their children \emph{and}-nodes. This nomenclature proposes the following reading of a proof-search space: a sequent is valid if and only if it is either vacuous (i.e., it is $\square$) or there is at least one or-node child all of whose children are valid.  The reading describes an algorithm for how one extracts particular reductions (i.e., attempts to find proof) from the space:

\begin{Definition}[Reduction] \label{def:reduction}
    	A finite subtree $\mathcal{R}$ of a proof-search space $\redop[L](S)$ is a reduction when it has root $S$ whose children are the reduction trees of all the children of one or-node following $S$ in the proof-search space. It is a successful reduction when all the leaves are $\square$.
\end{Definition}

\begin{Example}\label{ex:psscont}
    The explicit subtree of the displayed proof-search space for $\langle \rm p \fatsemi \rm \aunit \fatsemi \rm p \to \rm q, \rm q\rangle$ in Example \ref{ex:pss} is a reduction in that space. It is a successful reduction because the terminal nodes are all $\square$.
\end{Example}

\begin{Lemma} \label{lem:reduction}
    A tree of sequents is an $\system{L}$-proof of $S$ if and only if it is a successful reduction in $\redop[L](S)$.
\end{Lemma}
\begin{proof}
	Immediate by induction on the height of proofs and the definition of reduction tree --- see, for example, previous work by the authors \cite{Samsonschrift}.
\end{proof}

The reductive view of logic expresses the sense in which completeness follows from the proof-theoretic definition of a logic. The idea of using proof-search in this way is similar to the more familiar approach of a term-model construction in that it employs the proof-theoretic principles to determine the structure of the frame, but it is different in that it does not commit one to building a particular model, merely to witness that all the requisite behaviours are present and, crucially, that only those behaviours are present. Put succinctly, the idea is that reasoning (i.e., analyzing and determining why a sequent is valid) in a formal system is characterized by proof-search spaces; for example, proof-theoretic reasoning about the truth of a sequent in BI is captured by the proof-search space for the sequent with respect to a sequent calculus. 

How does the work on proof-search spaces apply to model-theoretic validity? The technology delivering the paper is that one can take the phrase \emph{model theory} literally; that is, one may study satisfaction in a frame as a theory of first-order classical logic by formalizing the implicit ambient logic in which mathematics is conducted as a meta-logic. This is the subject of Section \ref{sec:modeltheoryquaclassicaltheory}. Being classical, the meta-logic comes with its own well-understood and well-behaved proof theory, and one uses the construction of proof-search space with respect to the meta-calculus to characterize model-theoretic reasoning. The equivalence claim of the proof-theoretic and model-theoretic semantics of BI is then understood strongly: one shows that the proof-search spaces of the former contains the proof-search spaces of the latter.

\section{Model Theory \emph{qua} Classical Theory} \label{sec:modeltheoryquaclassicaltheory} In this section we capture BI-frames and satisfaction as a theory of classical logic, called the meta-logic, such that a validity judgment holds iff there is a formal proof of a certain meta-sequent. The section is composed of three parts: first we define the meta-logic and the encoding of semantic validity in Section \ref{subsec:encoding}; second, we develop a proof theory for the meta-logic in Section \ref{subsec:metaprooftheory}; and, third, we characterize \emph{reasoning} about validity for BI as captured by proof-search in the meta-logic in Section~\ref{subsec:reasoning}.

\subsection{Meta-logic}\label{subsec:encoding}
The definition of model-theoretic validity $\Gamma \models \phi$ takes the form of an implication (i.e., ...\emph{if} $w \satisfies \Gamma \emph{ then } w \satisfies \phi$.). This point of this observation is to call to attention that the mathematics used to study logic happens itself within a logic, which may be called \emph{the ambient logic}. In this section we formalize the ambient logic into a symbolic (meta-)logic that then allows previously developed techniques (e.g., the work of proof-search in Section \ref{sec:proofsearch}) to be applied to study model theory. This move is quite natural as the original impetus for logic was as a way to formulate a \emph{foundation of mathematics}, to understand symbolically the logic in which mathematical reasoning takes places; for example, one may use logic to study the natural numbers, set theory, etc. There is also a precedent in that many of the applications of logic proceed by formalizing some situtation symbolically; for example, logic may be used to model communications protocols, performed automated verification tasks, etc. From this perspective, the present paper differs only in that the application is to study another logic --- viz. BI. Indeed, the idea of formalizing one logic in another is already present in the literature --- see, for example, Gabbay \cite{Gabbay1993,Gabbay1996} and Negri \cite{Negri2005}.

The meta-logic introduced in this section is designed for the model theory of BI, and so it is determined by formulating the definitions of the previous section symbolically. In short, it is a two-sorted first-order classical logic. 

\begin{Definition}[Syntax of the Meta-logic]
    The world-terms are given by a set of world-variables $\set{V}_w$ together with a world-constant $e$ and a world-constant $\pi$. The formulas-terms $\set{T}_f$ are composed of formulas-variables and the grammar of BI-formulas in which connectives are regarded as function-symbols. 
    
    The set of atomic meta-formulas $\set{A}$ is defined as follows:
    \[
    \set{A} ::=\{ x \satisfies \phi, R(x,y,z), x \preceq y, x = y  \mid x,y,z \in \set{V}_w \cup \{e, \pi\} \text{ and } \phi \in \set{T}_f\}
    \]
    
    The set of meta-formulas, denoted $\set{M}$, is defined by the following grammar:
    \[
    \Phi ::= A \in \set{A} \mid \Phi \with \Phi \mid \Phi \parr \Phi \mid \Phi \Rightarrow \Phi \mid \forall \phi \Phi \mid \forall w \Phi \mid \exists w \Phi \mid \exists \phi \Phi \mid \square \mid \hash 
    \]
\end{Definition}

As above, $w \satisfies \Gamma$ abbreviates $w \satisfies \lfloor \Gamma \rfloor$. Aside from $\hash$, used as meta-falsum, it is assumed that the meaning of the connectives is clear; the choice to overload the symbol $\square$ as both the empty sequent and as meta-verum in this paper is purposeful, as discussed below. The notation $\Phi \iff \Psi$ abbreviates $(\Phi \Rightarrow \Psi) \with (\Psi \Rightarrow \Phi)$. Throughout, the symbols $\Sigma$ and $\Pi$ are reserved for lists of meta-formulas; $\Sigma \sim \Sigma'$ denotes that $\Sigma$ and $\Sigma'$ are permutations of each other. Finally, we call quantifier-free meta-formulas \emph{propositional} meta-formulas, and we call meta-atoms of the form $w \satisfies \phi$ \emph{assertions}.

The definitions of the previous section can be encoded in the meta-logic; that is, one may regard the model theory of BI \emph{qua} a theory in the meta-logic. There are two parts to capture: the sentences governing BI-frames $\Sigma_{\mathfrak{M}}$ (Definition \ref{def:bialg}) and sentences governing satisfaction $\Sigma_{\satisfies}$ (Definition \ref{def:satisfaction}).  

The sentences in $\Sigma_{\mathfrak{M}}$ are the universal closure of the following, in which $u,v,w,x,y,z$ are world-variables and $\phi$ is a formula variable:
\[
\begin{array}{c}
\underbrace{R(x,x,e)}_{\text{unitality}} \quad \underbrace{\big(R(x,y,z) \Leftrightarrow R(x,z,y) \big)}_{\text{commutativity}} \quad \underbrace{\big( w \preceq u \Rightarrow  ( w \satisfies \phi \Rightarrow u \satisfies \phi) \big)}_{\text{persistence}} \\   \underbrace{\big(R(x,w,y) \with R(y,u,v) \Rightarrow \exists z( R(x,z,v) \with R(z,w,u)) \big)}_{\text{associativity}} 
 \quad \underbrace{w = \pi \Rightarrow w \satisfies \phi}_{\text{absurdity}}
\end{array}
\]

The sentences in $\Sigma_{\satisfies}$ are given by the universal closure of the meta-formulas in Figure \ref{fig:sat:BI:symbolic} in which quantifiers are taken to be over each implicit conjunct separately, which merits comparison with Figure \ref{fig:sat:BI}. There are two significant differences between Figure \ref{fig:sat:BI} and Figure \ref{fig:sat:BI:symbolic}: first, there is no clause for $w \satisfies \rm p$,  where $\rm p \in \Atoms$; second, there is no clause for $w\satisfies \I$. This is an effort to simplify computations about satisfaction below. The elimination of a clause for atomic satisfaction follows from working with validity directly (i.e., without passing though truth-in-a-model) as interpretation are no longer required; that is, atomic satisfaction is captured by an atomic tautology, $(w \satisfies \rm p) \Rightarrow (w \satisfies \rm p)$. The justification of the non-presence of a $\I$-clause is postponed to the end of Section \ref{subsec:reasoning} as some additional technology is useful to facilitate discussion. 

\begin{figure}[t]
    \centering
    \fbox{
    \begin{minipage}{0.95\linewidth}
    \[
    \begin{array}{lcc}
        w \satisfies \top & \iff & \square \\
         w \satisfies \bot & \iff & w = \pi \\
         w \satisfies \phi \land \psi & \iff & (w \satisfies \phi) \with (w \satisfies \psi) \\
         w \satisfies  \phi \lor \psi & \iff & (w \satisfies \phi)  \parr (w \satisfies \psi) \\
         w \satisfies \phi \to \psi &  \iff & \forall u (w \preceq u \Rightarrow (u \satisfies \phi  \Rightarrow u \satisfies \psi )) \\
        w \satisfies \phi * \psi & \iff & \exists u,v : \, R(w,u,v) \with u \satisfies \phi \with v \satisfies \psi \\
         w \satisfies \phi \wand \psi & \iff & \forall u, w' (R(w',w,u)  \Rightarrow (u \satisfies \phi  \Rightarrow w' \satisfies \psi)) \\
    \end{array}
    \]
    \end{minipage}
    }
    \caption{Satisfaction for BI (Symbolic)}
    \label{fig:sat:BI:symbolic}
\end{figure} 

The union of $\Sigma_{\satisfies}$ and $\Sigma_{\mathfrak{M}}$ is denoted $\Sigma_{\logic{BI}}$. Let $[\Sigma]$ denote the meta-conjunction of all the meta-formulas in $\Sigma$. Under this encoding of BI's model theory, the assertion $\Gamma \models \phi$ may be understood as the following meta-implication:
\[
[\Sigma_{\logic{BI}}] \with (w \satisfies \Gamma) \Rightarrow (w \satisfies \phi)
\]
The significance of this is that all the familiar tools of classical logic become available, including sequent calculi for reasoning about when the above implication holds. The details of the proof-theoretic tools used for the meta-logic in this paper is reserved for Section \ref{subsec:metaprooftheory}, but sequents may be captured presently.

\begin{Definition}[Meta-sequent]
    A meta-sequent is a pair $\Sigma\,:\, \Pi $ in which $\Sigma$ and $\Pi$ are lists of meta-formulas. The empty pair, denoted $\square$, is also a sequent.
\end{Definition}

 The overloading of the symbol with meta-top $\metatop$ is not a problem as they are different kinds of objects: one is a meta-sequent, the other a meta-formula. We do not distinguish $ \Sigma, \metatop : \Pi $ and  $ \Sigma: \Pi $. A meta-sequent that captures (putative) semantic judgments are called \emph{basic validity sequents}:

\begin{Definition}[Basic Validity Sequent]
    A basic validity sequent (BVS) is a sequent of the following form:
    \[
    \langle \Sigma_{\logic{BI}}, (w \satisfies \Gamma) \,:\,  (w \satisfies \phi) \rangle
    \]
\end{Definition}
\begin{Definition}[Complex Validity Sequent]
    A complex validity sequent (CVS) is a sequent of the following form in which $\bar{\Sigma}$ and $\bar{\Pi}$ are sets of assertions:
    \[
\langle \Sigma_{\logic{BI}}, \bar{\Sigma} \,:\,  \bar{\Pi} \rangle
    \]
\end{Definition}

\subsection{Meta-logic Proof Theory}\label{subsec:metaprooftheory}

Having encoded the (putative) semantics as a theory of classical logic, all of the tools of classical logic become available. In particular, a BVS may be established simply by witnessing a proof for it in a proof system for classical logic (e.g., Gentzen's \cite{Gentzen} \system{LK}). We desire to establish a simulation between reasoning model-theoretically for validity and proof-theoretically for provability, both characterized by proof-search, hence we will restrict proof-search in the more flexible logic (i.e., the meta-logic). This restriction is the subject of this section, which contains all the technical aspects. The resulting calculus of validity used  to prove soundness and completeness is given in Section \ref{subsec:reasoning}.

The logic of BI is constructive. Consequently, one expects satisfaction to be constructive in the sense that, relative to $\Sigma_{\logic{BI}}$, if $w \satisfies \Gamma \Rightarrow w \satisfies \phi$ holds, then there should be a constructive proof of it. For this reason, we begin with a meta-sequent calculus for intuitionistic logic, which is based on Dummett's \cite{Dummett2000} multiple-conclusioned system.

\begin{Definition}[System $\system{DLJ}$]
    System $\system{DLJ}$ is composed of the rules in Figure \ref{fig:dlj} in which $\theta_X$ denotes a substitution for $X$ and $\hat{\theta}_X$ denotes a substitution for $X$ by an eigenvariable.
\end{Definition}

Rules for negation have been elided from Figure \ref{fig:dlj} as $\Sigma_{\logic{BI}}$ is negation-free so they will not be required at any point.

\begin{figure}[t]
    \fbox{
	\begin{minipage}{.95\textwidth}
		\centering 
		\vspace{0.2cm}
		$
		\infer[\lrn{\weak}]{\Phi,\Sigma \proves \Pi}{\Sigma\proves\Pi}
		\quad
		\infer[\rrn{\weak}]{\Sigma \proves \Pi,\Phi}{\Sigma\proves\Pi}
		\quad
		\infer[\lrn{\exch}]{\Sigma_1,\Phi,\Psi,\Sigma_2\proves \Pi}{\Sigma_1,\Psi,\Phi,\Sigma_2\proves \Pi}
		\quad
		\infer[\rrn{\exch}]{\Sigma\proves \Pi_1,\Phi,\Psi,\Pi_2}{\Sigma\proves \Pi_1,\Psi,\Phi,\Pi_2}
		\quad
		$
		\\[1.5ex]
		$
		\infer[\rrn\land]{\Sigma \proves \Pi,\Phi \land \Psi}{\Sigma \proves \Phi,\Pi & \Sigma\proves\Psi,\Pi}
		\mkern12mu
		\infer[\lrn\land]{\Phi \land  \Psi,\Sigma\proves \Pi}{\Phi , \Psi ,\Sigma\proves\Pi}
		\mkern12mu
		\infer[\rrn \to]{\Sigma \proves \Phi \to \Psi}{\Phi,\Sigma \proves \Psi}
		$
		\\[1.5ex]
		$
		\infer[\lrn{\lor}]{\Phi\lor\Psi,\Sigma \proves \Pi}{\Phi,\Sigma \proves\Pi & \Psi,\Sigma \proves\Pi }
	       \quad
		\infer[\rrn{\lor}]{\Sigma \proves \Pi , \Phi \lor  \Psi}{\Sigma \proves\Pi,\Phi , \Psi }
		\quad
		\infer[\lrn \to]{\Phi \to \Psi,\Sigma \proves \Pi}{\Sigma \proves \Pi,\Phi & \Psi,\Sigma \proves \Pi}
		$
		\\[1.5ex]
		$
		\infer[\lrn \forall]{\forall X\Phi, \Sigma \proves \Pi}{\Phi\theta_X,\Sigma \proves \Pi}
			\quad
		\infer[\rrn \forall]{\Sigma \proves \forall X \Phi}{\Sigma \proves \Phi\hat{\theta}_X}
		\quad
		\infer[\lrn \exists]{\exists X\Phi, \Sigma \proves \Pi}{\Phi\hat{\theta}_X,\Sigma \proves \Pi}
		\quad
		\infer[\rrn \exists]{\exists X\Phi, \Sigma \proves \Pi}{\Phi\theta_X,\Sigma \proves \Pi}
		$
		\\[1.5ex]
		$
			    \infer[\ax]{\Phi \proves \Phi}{\square}
\quad
	    \infer[\rrn \square]{\Phi \proves \square}{\square}
		$
	    \vspace{0.2cm}
    \end{minipage}
    }
    \caption{System \system{DLJ}}
    \label{fig:dlj}
\end{figure}

Of course, the encoding of Section \ref{subsec:encoding} is in a  \emph{classical} meta-logic and thus, despite the above intuition, we have no guarantee that \system{DLJ} is adequate for proving BVSs. To this end, it suffices to show that the following rules are admissible for $\system{DLJ}$-proofs of BVSs as including them recovers a meta-sequent calculus for classical logic:
      \[
   \infer[\rrn{\forall^\mathsf{K}}]{   \Sigma : \Pi, \forall x\Phi  }{ \Sigma : \Pi, \Phi  }
     \qquad
   \infer[\rrn{\Rightarrow^\mathsf{K}}]{ \Sigma : \Pi, \Phi \Rightarrow \Psi  }{  \Phi, \Sigma \proves \Pi, \Psi }
   \qquad
   \infer[\rrn{\cont}]{ \Sigma : \Pi, \Phi }{ \Sigma : \Pi, \Phi, \Phi}
      \qquad
   \infer[\lrn{\cont}]{ \Phi, \Sigma : \Pi, \Phi }{ \Phi, \Phi, \Sigma : \Pi}
   \]

Two rules are immediate:

\begin{Lemma}\label{lem:admissibilityofcontractionright}
   The $\rrn{\cont}$-rule and $\lrn \cont$-rule are admissible in $\system{DLJ}$.
\end{Lemma}
\begin{proof}
Immediate by the idempotency of IL disjunction and conjunction.
\end{proof}

The remaining two may be regarded as generalized versions of $\rrn \forall$ and $ \rrn \Rightarrow$, respectivel. Define $\system{DLJ}^\mathsf{K} : = \system{DLJ}\cup \{\rrn{\forall^\mathsf{K}}, \rrn{\Rightarrow^\mathsf{K}}, \rrn{\cont}, \lrn \cont\}$. The relationship between $\system{DLJ}$ and $\system{DLJ}^\mathsf{K}$ is the same as the relationship between Dummett's \cite{Dummett2000} (multiple-conclusioned) sequent calculus for intuitionistic logic and Gentzen's \cite{Gentzen} sequent calculus for classical logic --- viz. that certain rules in the former system are guarded by a single-conlusioned condition that if relaxed (or, \emph{generalized} --- see previous work by the authors \cite{Samsonschrift}) to be multiple-conclusioned recovers the latter system. Why can this guard be removed for proofs of BVSs without enlarging the space of valid meta-sequents? A sufficient guard is already captured by the change of world that is required when encountering an implicational formula in the extract of CVSs (i.e., an assertion of the form $w \satisfies \phi \to \psi$ or $w \satisfies \phi \wand \psi)$, meaning that the new assumption cannot interact with the other disjunctions. We shall return to this idea once a certain definitions and technical results have been given that will facilitate the discussion.

Consequently, we restrict proof-search in the meta-logic to consider only computations that begins with a CVS and produces a set of CVS that collectively justify the original one, these phases being identified by the appeal to a clause from the theory $\Sigma_{\logic{BI}}$. Such steps are called \emph{resolutions}.

\begin{Definition}[Resolution] \label{def:resolution}
    A resolution is a derivation that instantiates a clause from $\Sigma_{\satisfies}$. A resolution is closed if the head of the clause matches with an assertion already present in the meta-sequent and one removes the head in the non-axiom premiss. 
\end{Definition}

 Note that it is without loss of generality that an $\lrn \Rightarrow$ immediately applies to the active formula of an $\lrn \forall$ as the quantifier rule is invertible.

For clarity, a resolution is a derivation of the following form: 
        \begin{prooftree}
    \AxiomC{$\Sigma_{\logic{BI}},\Sigma \proves \Pi,\Phi$}
    \AxiomC{$\Sigma_{\logic{BI}}, \Psi, \Sigma \proves \Pi$}
    \RightLabel{$\lrn \Rightarrow$}
    \BinaryInfC{$\Sigma_{\logic{BI}}, \Phi \Rightarrow \Psi, \Sigma \proves \Pi$} 
    \RightLabel{$\lrn \forall$}
     \UnaryInfC{$\Sigma_{\logic{BI}}, \Sigma \proves \Pi$}
    \end{prooftree} 
   
\noindent And a \emph{closed} resolution is a derivation of either of the following forms:
    \begin{prooftree}
    \AxiomC{$\square$}
    \UnaryInfC{$\Sigma_{\logic{BI}}, \Phi, \Sigma \proves \Pi, \Phi$}
        \AxiomC{$\Sigma_{\logic{BI}}, \Psi, \Sigma \proves \Pi$}
        \RightLabel{$\lrn \weak$}
        \UnaryInfC{$\Sigma_{\logic{BI}}, \Psi, \Phi, \Sigma \proves \Pi$}
        \RightLabel{$\lrn \Rightarrow$}
    \BinaryInfC{$\Sigma_{\logic{BI}}, \Phi \Rightarrow \Psi, \Phi, \Sigma \proves \Pi$} 
    \RightLabel{$\lrn \forall$}
     \UnaryInfC{$\Sigma_{\logic{BI}}, \Phi, \Sigma \proves \Pi$}
    \end{prooftree}
    
        \begin{prooftree}
        \AxiomC{$\Sigma_{\logic{BI}}, \Sigma, \proves \Pi, \Phi$}
        \RightLabel{$\rrn \weak$}
    \UnaryInfC{$\Sigma_{\logic{BI}}, \Sigma, \proves \Pi, \Psi, \Phi$}
       \AxiomC{$\square$}
        \UnaryInfC{$\Sigma_{\logic{BI}}, \Psi, \Sigma \proves \Pi, \Psi$}
        \RightLabel{$\lrn \Rightarrow$}
    \BinaryInfC{$\Sigma_{\logic{BI}}, \Phi \Rightarrow \Psi, \Sigma \proves \Pi, \Psi$} 
    \RightLabel{$\lrn \forall$}
     \UnaryInfC{$\Sigma_{\logic{BI}}, \Sigma \proves \Pi, \Psi$}
    \end{prooftree}
    
When a resolution is closed, we may denote the reduction by the premiss that is not a tautology, labeling it by the name of the justifying clause; that is, let $\rn{cl.}$ be name of some clause from $\Sigma_{\satisfies}$ that instantiates to $\Phi \Rightarrow \Psi$, then the derivations above may be written as follows:\\
  
 \begin{minipage}{0.47\linewidth}
    \begin{prooftree}{}
        \AxiomC{$\Sigma_{\logic{BI}}, \Psi, \Sigma \proves \Pi$}
    \RightLabel{$\rn{cl}_\mathsf{L}$}
     \UnaryInfC{$\Sigma_{\logic{BI}}, \Phi, \Sigma \proves \Pi$}
    \end{prooftree}
    \end{minipage}
    \begin{minipage}{0.47\linewidth}
        \begin{prooftree}
     \AxiomC{$\Sigma_{\logic{BI}}, \Sigma \proves \Pi, \Phi$}
    \RightLabel{$\rn{cl}_\mathsf{R}$}
     \UnaryInfC{$\Sigma_{\logic{BI}}, \Sigma \proves \Pi, \Psi$} 
    \end{prooftree}
    \end{minipage}\\
    
Denoted in this way, resolutions may be thought of as rules (or, more precisely, reduction operators). This allows us to emphazise the steps that make use of the theory $\Sigma_{\logic{BI}}$ while de-emphasizing the meta-logical ones. 

\begin{Example} \label{ex:instimpl} The following is intuitively a reduction of $ \Gamma \models \phi \land \psi$ to $\Gamma \models \phi$ and $\Gamma \models \psi$ that begins by a resolution using the $\land$-clause and then uses $\rrn \with$, the tautology has been suppressed for readability:\\

\begin{minipage}{0.98\linewidth}
\begin{prooftree}
\AxiomC{$\Sigma_{\logic{BI}}, (w \satisfies \Gamma) \proves (w \satisfies \phi)$}
\AxiomC{$\Sigma_{\logic{BI}}, (w \satisfies \Gamma) \proves (w \satisfies \psi)$}
        \BinaryInfC{$\Sigma_{\logic{BI}}, (w \satisfies \Gamma) \proves (w \satisfies \phi) \with (w \satisfies \psi)$}
        \RightLabel{$\lrn \weak$}
        \UnaryInfC{$\Sigma_{\logic{BI}}, (w \satisfies \Gamma) \proves (w \satisfies \phi) \with (w \satisfies \psi), (w \satisfies \phi \land \psi)$}
             \AxiomC{$\square$}
             \RightLabel{$\ax$}
\UnaryInfC{}
    \RightLabel{$\rrn \Rightarrow$}
    \BinaryInfC{$\Sigma_{\logic{BI}}, (w \satisfies \phi \land \psi) \impliedby (w \satisfies \phi \with w \satisfies \psi), (w \satisfies \Gamma) \proves (w \satisfies \phi\land\psi)$}
    \RightLabel{$\lrn \forall$}
    \doubleLine
\UnaryInfC{$\Sigma_{\logic{BI}}, (w \satisfies \Gamma) \proves (w \satisfies \phi\land\psi)$}
\end{prooftree}
\end{minipage}\\

The same derivation may be denoted as follows:

\begin{prooftree}
\AxiomC{$(w \satisfies \Gamma) \proves (w \satisfies \phi)$}
    \AxiomC{$(w \satisfies \Gamma) \proves (w \satisfies \psi)$}
    \RightLabel{$\rrn \land$}
\BinaryInfC{$(w \satisfies \Gamma) \proves (w \satisfies \phi) \with (w \satisfies \psi)$}
\RightLabel{$\land$-clause}
\UnaryInfC{$(w \satisfies \Gamma) \proves (w \satisfies \phi\land\psi)$}
\end{prooftree}

\end{Example}

Since the theory $\Sigma_{\logic{BI}}$ may be taken to be conserved in the context, henceforth it may be suppressed without further comment.

It is reasoning by resolution that captures what it means \emph{to use} a clause of satisfaction, hence the sequent calculus for the meta-logic ought to have resolutions be the primary operational step during proof-search.  The fact that resolution is how semantic reasoning is conducted is not surprising; after all, that a theory composed of clauses may be used to \emph{define} a predicate is the idea underpinning the logic-based approach to artificial intelligence known as Logic Programming (LP) invented by Kowalski \cite{Kowalski1971,kowalski1979book}.

Resolutions can be used not only to perform computation about the satisfaction relation, but to break up the structure of bunches such that they may be read in the form of a classical context. We may think of this as \emph{unpacking} the bunch. Of course, it is essential that no information is lost in this process.

\begin{Definition}[Unpacking]
    An unpacking is a sequence of closed resolutions using $\land$- and $\ast$-clauses in the context with $\lrn \exists$ and $\lrn \with$ applied eagerly.
\end{Definition}
\begin{Example}
The following computation constitutes an unpacking:
\begin{scprooftree}{0.95}
\AxiomC{$R(w,x,y) , (x \satisfies \Gamma) , (y \satisfies \Delta) \with (y \satisfies \Delta' \fatsemi \Delta'), (u \satisfies \Gamma') \proves (w \satisfies \phi), (u \satisfies \psi) $}
\RightLabel{$\land$-clause}
\UnaryInfC{$R(w,x,y) , (x \satisfies \Gamma) , (y \satisfies \Delta \fatsemi \Delta'), (u \satisfies \Gamma') \proves (w \satisfies \phi), (u \satisfies \psi) $}
\doubleLine
\RightLabel{$\lrn \with$}
\UnaryInfC{$R(w,x,y) \with (x \satisfies \Gamma) \with (y \satisfies \Delta \fatsemi \Delta'), (u \satisfies \Gamma') \proves (w \satisfies \phi), (u \satisfies \psi) $}
\RightLabel{$\lrn \exists$}
\UnaryInfC{$\exists x,y (R(w,x,y) \with x \satisfies \Gamma \with y \satisfies \Delta \fatsemi \Delta'), (u \satisfies \Gamma') \proves (w \satisfies \phi), (u \satisfies \psi) $}
\RightLabel{$\ast$-clause}
\UnaryInfC{$(w \satisfies \Gamma\fatcomma (\Delta \fatsemi \Delta')), (u \satisfies \Gamma') \proves (w \satisfies \phi), (u \satisfies \psi) $}
\end{scprooftree}
\end{Example}

The notation $\Sigma_{w , \Gamma}$ denotes a theory that arises from an unpacking of $w \satisfies \Gamma$. Unpackings do not have to be total; that is, one can have $w \satisfies \Gamma(\phi)$ unpack to a theory $\Sigma_{w , \Gamma(\phi)}$ containing a meta-formula $x \satisfies \phi$. In this case, the theory may be denoted $\Sigma_{w , \Gamma(\phi), x \satisfies\phi}$. It is partial in the sense that the unpacking does not continue on the assertion $w \satisfies \phi$,

\begin{Lemma}[Packing] \label{lem:packing}
   Both packing and unpackings are invertible.
\end{Lemma}
\begin{proof}
The result follows from the invertibility of $\lrn{ \with}$ and $\lrn \exists$, as witnessed by the following computations:
\begin{prooftree}
\AxiomC{$ \Sigma, (w\satisfies \phi \land \psi) \proves \Pi $}
\RightLabel{$\land$-clause}
\UnaryInfC{$ \Sigma, (w\satisfies \phi) \with (w \satisfies \psi) \proves \Pi $}
\dashedLine
\RightLabel{$\rrn{\with}^{-1}$}
\UnaryInfC{$\Sigma, (w\satisfies \phi), (w \satisfies \psi) \proves \Pi $}
\RightLabel{$\rrn \with$}
\UnaryInfC{$\Sigma, (w\satisfies \phi) \with (w \satisfies \psi) \proves \Pi $}
\RightLabel{$\land$-clause}
\UnaryInfC{$\Sigma, (w\satisfies \phi \land \psi) \proves \Pi $}
\end{prooftree}
and
\begin{prooftree}
\AxiomC{$ \Sigma, (w\satisfies \phi * \psi) \proves \Pi $}
\RightLabel{$\ast$-clause}
\UnaryInfC{$ \Sigma, \exists u,v(R(w,u,v) \with u \satisfies \phi \with v \satisfies \psi) \proves \Pi $}
\dashedLine
\RightLabel{$\lrn{\exists}^{-1}$}
\UnaryInfC{$ \Sigma, R(w,u,v) \with (u \satisfies \phi) \with (v \satisfies \psi) \proves \Pi $}
\dashedLine\doubleLine
\RightLabel{$\rrn{\with}^{-1}$}
\UnaryInfC{$ \Sigma, R(w,u,v), (u \satisfies \phi), (v \satisfies \psi) \proves \Pi $}
\doubleLine
\RightLabel{$\rrn \with$}
\UnaryInfC{$\Sigma, R(w,u,v) \with u \satisfies \phi \with v \satisfies \psi \proves \Pi $}
\RightLabel{$\lrn \exists$}
\UnaryInfC{$ \Sigma, \exists u,v(R(w,u,v) \with u \satisfies \phi \with v \satisfies \psi) \proves \Pi $}
\RightLabel{$\ast$-clause}
\UnaryInfC{$\Sigma, (w\satisfies \phi * \psi) \proves \Pi $}
\end{prooftree}

\end{proof}

Recall that the reason one expects $\system{DLJ}$ to be adequate for reasoning about BI-validity is because the guard distinguishing $\rrn{\Rightarrow}$ and $\rrn{\forall}$ from $\rrn{\Rightarrow^\mathsf{K}}$ and $\rrn{\forall^\mathsf{K}}$  is captured by the change-of-world when encountering an implication formula in the extract of a CVS. This idea is witnessed in the following example:

\begin{Example} \label{ex:lemfails}
To see how the change of world acts as a sufficient guard for BI-validity to be constructive, we may see how $\system{DLJ}^\mathsf{K}$ avoids the law of the excluded middle from holding in BI,

\begin{scprooftree}{0.95}
\AxiomC{$(w \satisfies \munit), (u \satisfies \munit\fatsemi \phi) \proves (w \satisfies \phi),  (u \satisfies \bot)$}
\RightLabel{$\land$-clause}
\UnaryInfC{$(w \satisfies \munit), (u \satisfies \munit), (u \satisfies \phi) \proves (w \satisfies \phi),  (u \satisfies \bot)$}
\RightLabel{$\rn{pers.}$}
\UnaryInfC{$(w \satisfies \munit), u \preceq w, (u \satisfies \phi) \proves (w \satisfies \phi),  (u \satisfies \bot)$ }
\RightLabel{$\rrn{\Rightarrow^\mathsf{K}}$}
\doubleLine
\UnaryInfC{$(w \satisfies \munit) \proves (w \satisfies \phi), (w \preceq u \Rightarrow (u \satisfies \phi \Rightarrow u \satisfies \bot))$}
\RightLabel{$\rrn{\forall^\mathsf{K}}$}
\UnaryInfC{$(w \satisfies \munit) \proves (w \satisfies \phi), \forall u (w \preceq u \Rightarrow (u \satisfies \phi \Rightarrow u \satisfies \bot))$}
\RightLabel{$\to$-clause}
\UnaryInfC{$(w \satisfies \munit) \proves (w \satisfies \phi), (w \satisfies (\phi \to \bot)) $}
\RightLabel{$\rrn \parr$}
\UnaryInfC{$(w \satisfies \munit) \proves (w \satisfies \phi) \parr (w \satisfies (\phi \to \bot))$}
\RightLabel{$\lor$-clause}
\UnaryInfC{$(w \satisfies \munit) \proves (w \satisfies \phi \lor (\phi \to \bot))$}
\end{scprooftree}
\noindent Moving to $u$ and using persistence means that one has all the contextual information about $w$ available (i.e., that $w \satisfies \Gamma$ is in the context enables $ u \satisfies \Gamma$ to be assumed); but since $u \satisfies \phi$ in the context and $w \satisfies \phi$ in the extract are different atoms since $u$ and $w$ are distinct, one has not reached an axiom. In short despite working in a classical system, suppressing an additional computational step, the above calculation witnesses that $\munit \models \phi \lor \phi \to \bot$ if $\munit \models \phi$ or $\phi \models \bot$, which is what one would expect of entailment for a constructive logic such as BI; that is, one knows that $\phi \lor \phi \to \bot$ is BI-valid only if one already knows that $\phi$ is BI-valid or one already knows that $\phi$ is BI-absurd.
\end{Example}

In particular, the change-of-guard means that a CVS to which one has reduces contains two \emph{independent} claims about validity, as witnessed in Example \ref{ex:lemfails}.

\begin{Definition}[World Independent]
    Sets of meta-formulas $\Sigma$ and $\Sigma'$ are world-indep\-endent if no free world-variable appearing in one appears in the other.
\end{Definition}

\begin{Lemma} \label{lem:worldindependence}
   Let $\Sigma, \Sigma', \Pi,\Pi'$ be sets of propositional meta-formulas such that $\Sigma, \Pi$ and $\Sigma', \Pi'$ are world-independent, then
  \[
   \Sigma_{\logic{BI}}, \Sigma, \Sigma' \proves \Pi, \Pi' \quad \text{iff} \quad \Sigma_{\logic{BI}}, \Sigma \proves \Pi \text{ or } \Sigma_{\logic{BI}}, \Sigma' \proves \Pi'
   \]
\end{Lemma}

Recall that $\Sigma \proves \Pi$ holds iff there there is a $\system{DLJ}^\mathsf{K}$-proof of $\Sigma:\Pi$, so the lemma may be proved by induction on $\system{DLJ}^\mathsf{K}$-proofs.

\begin{proof}
The \emph{if} direction follows immediately by $\lrn{\weak}$ and $\rrn{\weak}$. For the $\emph{only if}$ direction suppose $\Sigma_{\logic{BI}}, \Sigma, \Sigma' \proves \Pi, \Pi'$, then there is a $\system{DLJ}^\mathsf{K}$-proof $\mathcal{D}$ of it. We proceed by induction the number of resolutions in such a proof.

\textsc{Base Case.} If $\mathcal{D}$ contains no resolutions, then $\Sigma_{\logic{BI}}, \Sigma, \Sigma' \proves \Pi, \Pi'$ is proved by $\ax$ together with the rules for the meta-connectives. But then there are proofs for $\Sigma_{\logic{BI}}, \Sigma \proves \Pi$ or $\Sigma_{\logic{BI}}, \Sigma' \proves \Pi'$ since the rules for the connectives cannot affect what world- or formula-variables.

\textsc{Induction Step.} If any resolution of  $\Sigma_{\logic{BI}}, \Sigma, \Sigma' \proves \Pi, \Pi'$ yields a meta-sequent of the form of the statement of the lemma, then the lemma follows immediately from the induction hypothesis. We show that this is the case.

The only non-obvious case is in the case of a closed resolution using the $\to$-clause or $\wand$-clause in the extract because they have universal quantifiers that would allow one to produce a meta-atom in the extract that contains both a world from $\Sigma, \Pi$ and $\Sigma', \Pi'$ simultaneously, thereby breaking world-independence. We show the $\to$-case, the other being similar. 

Let $\Sigma = \Sigma'', w \satisfies \phi \to \psi$ and suppose $u$ is a world variable appearing in $\Sigma', \Pi'$, then we have the following computation in which $\Sigma_{\logic{BI}}$ has been suppressed in the context:

\begin{scprooftree}{0.95}
\AxiomC{$ \Sigma'', \Sigma' \proves \Pi, \Pi', w \prec u$}
\AxiomC{$ \Sigma'', \Sigma' \proves \Pi, \Pi', (u \satisfies \phi) $}
\AxiomC{$ \Sigma'', \Sigma', (u \satisfies\psi) \proves \Pi, \Pi' $}
\doubleLine
\RightLabel{$\lrn{\Rightarrow}^\mathsf{K}$}
\TrinaryInfC{$ \Sigma'', (w \prec u \Rightarrow (u \satisfies \phi \Rightarrow u \satisfies \psi)), \Sigma' \proves \Pi, \Pi'$}
\RightLabel{$\lrn \forall$}
\UnaryInfC{$\Sigma'', \forall x(w \prec x \Rightarrow (x \satisfies \phi \Rightarrow x \satisfies \psi)), \Sigma' \proves \Pi, \Pi'$}
\RightLabel{$\to$-clause}
\UnaryInfC{$\Sigma'', (w \satisfies \phi \to \psi), \Sigma' \proves \Pi, \Pi'$}
\end{scprooftree}

The $w \prec u$ may be removed from the leftmost premiss because the only way for the meta-atom to be used in the remainder of the proof is if $w \prec u$ appears in the context, but this is impossible. Hence, without loss of generality, $\mathcal{D}$ applies $\rrn \weak$ to the branch, yielding $\Sigma'', \Sigma' \proves \Pi, \Pi'$. The result follows by induction hypothesis. 

\end{proof}

To prove that \system{DLJ} is adequate for proofs of CVSs it only remains to argue that the change-of-world guard is implemented whenever it is required, and that it indeed results in a world-independent situation.
 
\begin{Lemma}\label{lem:admissabilityofclassicalimplication}
   The  $\rrn{\forall^\mathsf{K}}$ and $\rrn{\Rightarrow^\rn{K}}$ rules are admissible for \system{DLJ}-proofs of CVSs:
\end{Lemma}
\begin{proof}
By case analysis on $\Sigma_{\logic{BI}}$, the only way for either conclusion to arise in a proof of a CVS is as a closed resolution of an implicational assertion (i.e., an assertion $w \satisfies \phi \multimap \psi$ in which either $\multimap \in \{\to, \wand\}$). In the case of $\to$-clause, without loss of generality, the resolution may be taken to be required for the proof such that persistence is applied eventually to $w \prec u$.  By permuting resolutions, we may assume that it is used immediately. Thus, by Lemma \ref{lem:packing}, one has the following computation: 

\begin{prooftree}
\AxiomC{$ \bar{\Sigma}, (w \satisfies \Gamma), (u \satisfies \Gamma \fatsemi \phi) \proves \bar{\Pi}, (u \satisfies \psi)$}
\RightLabel{$\land$-clause}
\UnaryInfC{$\bar{\Sigma}, (w \satisfies \Gamma), (u \satisfies \Gamma), (u \satisfies \phi) \proves \bar{\Pi}, (u \satisfies \psi)$}
\RightLabel{$\rn{pers.}$}
\UnaryInfC{$\bar{\Sigma}, (w \satisfies \Gamma) , w \preceq u, (u \satisfies \phi) \proves \bar{\Pi}, (u \satisfies \psi)$}
\doubleLine
\RightLabel{$\rrn{\Rightarrow^\mathsf{K}}$}
\UnaryInfC{$\bar{\Sigma}, (w \satisfies \Gamma ) \proves \bar{\Pi}, (w \preceq u \Rightarrow (u \satisfies \phi \Rightarrow u \satisfies \psi)$}
\RightLabel{$\rrn{\forall^\mathsf{K}}$}
\UnaryInfC{$\bar{\Sigma}, (w \satisfies \Gamma)  \proves \bar{\Pi}, (w \preceq u \Rightarrow (u \satisfies \phi \Rightarrow u \satisfies \psi))$}
\end{prooftree}\vspace{0.2cm}

\noindent In the case of the $\wand$-clause, by Lemma \ref{lem:packing}, one has the following derivation

\begin{prooftree}
\AxiomC{$\bar{\Sigma}, (w \satisfies \Gamma ), (w' \satisfies \Gamma \fatcomma \psi) \proves \bar{\Pi}, (w' \satisfies \psi)$}
\RightLabel{$\ast$-clause}
\UnaryInfC{$\bar{\Sigma}, (w \satisfies \Gamma ), R(w',w,u), u \satisfies \phi \proves \bar{\Pi}, (w' \satisfies \psi)$}
\doubleLine
\RightLabel{$\rrn{\Rightarrow^\mathsf{K}}$}
\UnaryInfC{$\bar{\Sigma}, (w \satisfies \Gamma) \proves \bar{\Pi} ,(R(w',w,u) \Rightarrow (u\satisfies \phi \Rightarrow w' \satisfies \psi)) $}
\RightLabel{$\rrn{\forall^\mathsf{K}}$}
\UnaryInfC{$ \bar{\Sigma}, (w \satisfies \Gamma) \proves \bar{\Pi},\forall w', u (R(w',w,u) \Rightarrow (u\satisfies \phi \Rightarrow w' \satisfies \psi)) $}
\end{prooftree}\vspace{0.2cm}

In either case, by the eigenvariable condition on universal instantiations, the premiss is a meta-sequent of the form $\Sigma_{\logic{BI}}, \Sigma ,\Sigma' \proves \Pi, \Pi'$ in which $\Sigma, \Pi$ and $\Sigma', \Pi'$ are world-independent. Hence, by Lemma \ref{lem:worldindependence}, the rules are admissible. 
\end{proof}

\begin{Lemma} \label{lem:compdlj}
A CVS holds iff it admits a \system{DLJ}-proof.
\end{Lemma}
\begin{proof}
Immediate by Lemma \ref{lem:admissabilityofclassicalimplication} and Lemma \ref{lem:admissibilityofcontractionright}.
\end{proof}

There remains a particular behaviour that is useful to eliminate from the calculus: in \system{DLJ} one may begin with a CVS and instantiate a meta-formula in $\Sigma_{\logic{BI}}$ with a world not present in meta-sequent, but such a world-variable represents an arbitrary world alien to information about models available in the sequent and therefore, intuitively, it cannot be a required part of the reasoning used to establish or refute the veracity of the sequent.  

\begin{Example} \label{ex:world-labourious}
The following derivation is a reduction of a BVS that begins with a resolution introducing a world alien to the original meta-sequent,
\begin{prooftree}
\AxiomC{$\Sigma_{\logic{BI}}, (w\satisfies \rm p \land \rm q) \proves (u \satisfies \top) $}
\AxiomC{$\Sigma_{\logic{BI}}, (w \satisfies \rm p \land \rm q) \proves (w \satisfies \rm p \lor \rm q) $}
\RightLabel{$\lrn \Rightarrow$}
   \BinaryInfC{$\Sigma_{\logic{BI}}, (u \satisfies \top) \Rightarrow \square, (w \satisfies \rm p \land \rm q) \proves (w \satisfies \rm p \lor \rm q)$}
   \RightLabel{$\lrn \forall$}
   \UnaryInfC{$\Sigma_{\logic{BI}}, (w \satisfies \rm p \land \rm q) \proves (w \satisfies \rm p \lor \rm q)$}
\end{prooftree}
\end{Example}

We eliminate computation such as in Example \ref{ex:world-labourious} so that after resolutions way may always interpret meta-sequents as BI-sequents (see Section \ref{subsec:reasoning}, below).

\begin{Definition}[World-conservative]
    A $\system{DLJ}$-proof of a CVS is said to be world-conser\-vative if in any instance of $\lrn \forall$ or $\rrn \exists$, every world-variable occurring in the premiss occurs in the conclusion.
\end{Definition}

\begin{Lemma} \label{lem:compdljwc}
If a CVS is holds iff it admits a world-conservative $\system{DLJ}$-proof.
\end{Lemma}

\begin{proof}
Since $\lrn \forall$ has no pre-conditions, the result follows by renaming variable; that is, if an inference 
\begin{prooftree}
\AxiomC{$ \Sigma_{\logic{BI}}, \Sigma, \Psi\theta_u \proves \Pi $}
   \UnaryInfC{$ \Sigma_{\logic{BI}}, \Sigma, \forall u \Psi \proves \Pi $}
\end{prooftree}
is not world-conservative (i.e., $\theta_u : u \mapsto x$ and $x$ does not appear in $\Sigma$ or $\Pi$), it can be made world-conservative by replacing all hereditary occurrences of $x$ in the proof by a world-variable $y$ that does appear in either $\Sigma$ or $\Pi$ --- for example, the above inference becomes the following, where $\theta_u:u \mapsto y$:
\begin{prooftree}
\AxiomC{$ \Sigma_{\logic{BI}}, \Sigma, \Psi\theta_u' \proves \Pi $}
   \UnaryInfC{$ \Sigma_{\logic{BI}}, \Sigma, \forall u \Psi \proves \Pi $}
\end{prooftree}

\end{proof}

\subsection{Reasoning about BI-validity} \label{subsec:reasoning}
The approach to completeness in this paper proceeds by witnessing a correspondence between reasoning about provability using BI's proof theory and reasoning about validity by using BI's model theory, both characterized as reductive proof-search. In Section \ref{subsec:metaprooftheory}, we restricted the proof-search space to enable the comparison. In this section, we introduce a calculus that embodies these restrictions, thereby enabling a clearer comparison between proof-search for BI-provability and BI-validity.

Observe that the world-variables in BVSs do not stand for particular worlds, but rather are generic representatives of worlds. These may be called \emph{eigenworlds}.  The context of the meta-sequent instantiates some basic information that restricts the set of BI-models in which the extract is supposed to holds; in particular, $\Sigma_{\logic{BI}}$ captures the \emph{concept} of a BI-frames, rather than providing an \emph{actual} frame. Consequently, when working with meta-sequents one is working with validity directly, bypassing truth-in-a-model.

\begin{Example} 
Suppose one has the meta-sequent $w \satisfies \rm r \proves w \satisfies \rm p * \rm q$ in $\rm p, \rm q,$ and $\rm r$ are propositional letters of BI. The relevant half of the $*$-clause is, including all quantifiers, the following:
\[
\forall \phi, \psi \forall x (\exists y,z (R(x,y,z) \with x \satisfies \phi \with y \satisfies \psi \with z \satisfies \psi) \Rightarrow  x \satisfies \phi*\psi)
\]

Resolving with this clause produces the following meta-sequent:
\[
w \satisfies \rm r \Rightarrow \exists y,z (R(w,y,z) \with y \satisfies \rm p \with z \satisfies \rm q) 
\]

In the absence of any specific worlds, one introduces eigenworlds $u$ and $v$ to eliminate the existential quantifiers for $y$ and $z$, respectively, yielding the following:
\[
w \satisfies \rm r \Rightarrow R(w,u,v) \with u \satisfies \rm p \with v \satisfies \rm q 
\]

All of this work has been done generically since $w$, $u$ and $v$ are devoid of features, and thus we know that the same reasoning can take place at any world in any model. Explicitly, suppose one were given an actual model $\model$, then the above shows that if it holds for actual worlds $a,b,c$  in $\model$ that $R(c, a, b)$, $a \satisfies \rm p$, and $b \satisfies \rm q$ hold, then \emph{necessarily} $c \satisfies \rm p * \rm q$ holds in $\model$. 
\end{Example}

With the concept of eigenworlds, the format of the calculus for validity in BI below can be understood. Each rule in it is a rule that captures precisely how a clause affects validity.

\begin{Definition}[System \system{VBI}]
    System $\system{VBI}$ is composed of the rules in Figure \ref{fig:vbi} in which the theory $\Sigma_{\logic{BI}}$ has been suppressed in the context and $\cl{\asso}$ is invertable.
\end{Definition}

\begin{figure}[t]
    \fbox{
    \scalebox{0.9}{
	\begin{minipage}{\textwidth}
		\centering 
		\vspace{0.2cm}
		$
		 \infer[\lcl \land]{ w \satisfies \Gamma(\phi\land\psi)  : w \satisfies \chi}{w \satisfies \Gamma(\phi\fatsemi\psi) : w \satisfies \chi}
		 \quad
		\infer[\rcl \land]{w \satisfies \Gamma : w \satisfies \phi \land \psi}{w \satisfies \Gamma : w \satisfies \phi  & w \satisfies \Gamma : w \satisfies \psi}
		$
		\\[1.5ex]
		 $    
		 \infer[\lcl \ast]{ w \satisfies \Gamma(\phi*\psi) : w \satisfies \chi }{ w \satisfies \Gamma(\phi\fatcomma\psi) : w \satisfies \chi}
		 \quad
		     \infer[\rcl \ast]{w \satisfies \Gamma \fatsemi (\Delta_1 \fatcomma \Delta_2) : w \satisfies \phi_1 * \phi_2}{w \satisfies \Delta_1 : w \satisfies \phi_1  & w \satisfies \Delta_2 : w \satisfies \phi_2}
		  $   
		  \\[1.5ex]
		 $ 
		  \infer[\lcl \lor]{ w \satisfies \Gamma(\phi\lor\psi) : w \satisfies \chi}{w \satisfies \Gamma(\phi) : w \satisfies \chi & w \satisfies \Gamma(\psi) : w \satisfies \chi}
		  \quad
		 	\infer[\rcl \lor]{w \satisfies \Gamma : w \satisfies \phi_1 \lor \phi_2}{w \satisfies \Gamma : w \satisfies \phi_i}
		 $
		 \\[1.5ex]
		 $
		 \infer[\lcl \to]{w \satisfies \Delta \fatsemi \phi \to \psi : w\satisfies \chi}{w \satisfies \Delta : w \satisfies \phi & w \satisfies \Gamma(\Delta,\psi) : w \satisfies \chi }
		 \quad
		  \infer[\rcl \to]{w \satisfies \Gamma : w \satisfies \phi \to \psi}{w \satisfies \Gamma \fatsemi \phi: w \satisfies \psi}
		  $
		  \\[1.5ex]
        $
\infer[\lcl \wand]{ w \satisfies \Gamma(\Delta_1 \fatcomma \Delta_2, \phi\wand\psi) : w \satisfies \chi }{ w \satisfies \Delta_2 : w \satisfies \phi & w \satisfies \Gamma(\Delta_1 \fatcomma \psi) : w \satisfies \chi}
\quad
       \infer[\rcl \wand]{w \satisfies \Gamma : w \satisfies \phi \wand \psi}{w \satisfies \Gamma \fatcomma \phi : w \satisfies \psi}
        $
		 \\[1.5ex]
		 $
		 \infer[\lcl \I]{w \satisfies \Gamma(\Delta \fatcomma \munit) : w\satisfies \chi }{w \satisfies \Gamma(\Delta) : w\satisfies \chi}
                \quad
                 \infer[\lcl \I]{w \satisfies \Gamma(\Delta) : w \satisfies\chi }{ w \satisfies \Gamma(\Delta \fatcomma \munit) : w \satisfies \chi}
                   $
		 \\[1.5ex]
		 $ 
		 \infer[\lcl \top]{w \satisfies \Gamma(\Delta) : w \satisfies \chi }{w \satisfies \Gamma(\Delta\fatsemi \aunit) : w \satisfies \chi}
		 \quad
		       \infer[\rcl \top]{w \satisfies \Gamma : w \satisfies \top}{\square}
      $
      \\[1.5ex]
      $
                \infer[\lcl \bot]{w \satisfies \Gamma(\bot) : w\satisfies \chi }{ w \satisfies \Gamma(\phi) : w\satisfies \chi }   
                \quad
                 \infer[\ax]{w \satisfies \Gamma\fatsemi \phi : w\satisfies \phi}{ \square }
        $
        \\[1.5ex]
        $
        \infer[\cl \comm]{w \satisfies \Gamma(\Delta_1 \fatcomma \Delta_2) : w \satisfies \chi}{ w \satisfies \Gamma(\Delta_2 \fatcomma \Delta_1) : w \satisfies \chi }
        \quad 
        \infer[\cl \asso]{w \satisfies \Gamma(\Delta_1 \fatcomma (\Delta_2 \fatcomma \Delta_3)) : w \satisfies \chi}{ w \satisfies \Gamma((\Delta_1 \fatcomma (\Delta_2) \fatcomma \Delta_3) : w \satisfies \chi }
        $
        \\[1.5ex]
        $
        \infer[\exch_\rn{1}]{w \satisfies \Gamma(\Delta_1 \fatsemi \Delta_2) : w \satisfies \chi}{ w \satisfies \Gamma(\Delta_2 \fatsemi \Delta_1) : w \satisfies \chi }
        \quad 
        \infer[\exch_\rn{2}]{w \satisfies \Gamma(\Delta_1 \fatsemi (\Delta_2 \fatsemi \Delta_3)) : w \satisfies \chi}{ w \satisfies \Gamma((\Delta_1 \fatsemi \Delta_2) \fatsemi \Delta_3) : w \satisfies \chi }
        $
        \\[1.5ex]
         $
        \infer[\cont]{w \satisfies \Gamma(\Delta) : w \satisfies\chi }{ w \satisfies \Gamma(\Delta \fatsemi \Delta) : w \satisfies \chi}
        \quad
        \infer[\weak]{w \satisfies \Gamma(\Delta \fatsemi \Sigma) : w \satisfies \chi}{ w \satisfies \Gamma(\Delta) : w \satisfies \chi }
        $
        \\[1.5ex]
        $
        \infer[\cut]{w \satisfies \Gamma(\Delta) : w \satisfies\chi }{ w \satisfies \Delta : w \satisfies \phi & w \satisfies \Gamma(\phi) : w \satisfies \chi}
        $
        
	    \vspace{0.2cm}
    \end{minipage}
    }
    }
    \caption{System \system{VBI}}
    \label{fig:vbi}
\end{figure}

\begin{Theorem} \label{thm:vbi}
 A BVS is valid iff it admits a $\system{VBI}$-proof.
\end{Theorem}
\begin{proof}
The soundness of \system{VBI} is immediate by observing that each rule follows as the application of a meta-formula in $\Sigma_{\logic{BI}}$; for example, the admissibility of $\rcl{\land}$ is witnessed in this way in Example \ref{ex:instimpl}.

It remains to argue for the completeness of $\system{VBI}$. By Lemma \ref{lem:compdljwc}, a BVS holds only if it admits a world-conservative $\system{DLJ}$-proof. But since $\system{DLJ}$ is an intuitionistic calculus, we have the same result for the single-conlusioned variant $\system{GLJ}$ (i.e., Gentzen's \cite{Gentzen} sequent calculus for intuitionstic logic). We proceed by case analysis on the possible proof-searches for the BVS in $\system{GLJ}$.

A proof-search  for a BVS in $\system{GLJ}$ is composed of steps from six classes: using an axiom, open resolutions, extract-closed resolutions, context-closed resolutions, using a frame law, and using a structural rules. We show that to each step there is a corresponding  reduction in $\system{VBI}$ with the same start and end meta-sequents. Hence, proof-searches in $\system{GLJ}$ as a whole correspond to proof-searches in $\system{VBI}$.

 Without loss of generality, each proof-search  begins with an unpacking of the BVS. We may write $\Pi_{w , \Gamma(\Delta), x}, x \satisfies \Delta$ to denote a theory $\Sigma_{w , \Gamma(\Delta),x\satisfies \Delta}$. Moreover, in the closed-resolution cases, we assume that the resolvant is immediately decomposed (i.e., it is principal in the next reduction), as otherwise the resolution could have been postponed until this is the case. We denote that there is a reduction taking $C$ to $P_1,...,P_n$ as follows:
 \[
 \infer[\Uparrow]{C}{P_1 & ... & P_n}
 \]
 In particular, if the proof-search continues by taking $P_1$,...$P_n$ to $P_1',...,P_n'$ respectively, the total reduction is captured as follows:
 \[
 \infer[\Uparrow]{C}{P_1' & ... & P_n'}
 \]

 By Lemma \ref{lem:packing}, in each case we apply a packing eagerly (i.e., whenever packing is applicable and results in a sequent different from the original). Of course, there are more than one possible ways to pack a sequent according to the order in which meta-atoms are combined. Reductions that begin with unpacking a sequent and then packing it in a different order are as follows:
 \[
 \infer[\Uparrow]{w \satisfies\Gamma(\Delta_1 \fatsemi (\Delta_2 \fatsemi \Delta_2)) : w \satisfies \chi}{
    \infer[\Uparrow]{\Pi_{w ,\Gamma(\Delta_1 \fatsemi (\Delta_2 \fatsemi \Delta_2),x}, (x \satisfies \Delta_1), (x \satisfies \Delta_2), (x \satisfies \Delta_3) : w \satisfies \chi}{
        w \satisfies\Gamma((\Delta_1 \fatsemi \Delta_2) \fatsemi \Delta_2 : w \satisfies \chi
        }
    }
 \]
 Indeed, the computation is invertable. These reductions are captured by $\system{VBI}$ as an instance of  $\exch_\mathsf{2}$.
 
\bigskip \noindent \textsc{Axiom} \newline
    System $\system{GLJ}$ contains two axioms:  $\ax$ and $\metatop$. Only of them is applicable to the unpacking of a BVS --- viz. $\ax$. If the proof-search used $\ax$, then the unpacking of the BVS was of the form $ \Sigma_{w , \Gamma}, w \satisfies \phi : w \satisfies \phi \rangle$. This is only possible if the BVS was of the form $w \satisfies \Gamma,\phi : w\satisfies\phi$. These reductions are captured by $\system{VBI}$ as an instance of $\rn{id}$.

\bigskip \noindent \textsc{Open Resolutions} \newline
The sequent $w \satisfies \Gamma(\Delta): w \satisfies \phi$ is unpacked to $\Sigma_{w , \Gamma(\Delta), x \satisfies \Delta } : w \satisfies \phi$.  The open resolution gives the following reduction in which either $\Phi = (x \satisfies \chi)$ or $\Psi = (x \satisfies \chi)$:
\[
\infer[\Uparrow]{\Sigma_{w , \Gamma(\Delta), x \satisfies \Delta } : w \satisfies \phi}{\Sigma_{w , \Gamma(\Delta), x \satisfies \Delta} : \Phi & \Sigma_{w , \Gamma(\Delta), x \satisfies \Delta}, \Psi : w \satisfies \phi }
\]
Without loss of generality, by the invertability of the requisite resolutions, each branch is continued with a closed-resolution such that both $\Phi$ and $\Psi$ become $(x \satisfies \chi)$. 
\[
\infer[\Uparrow]{\Sigma_{w , \Gamma(\Delta), x \satisfies \Delta } : w \satisfies \phi}{\Sigma_{w , \Gamma(\Delta), x \satisfies \Delta} : x \satisfies \chi & \Sigma_{w , \Gamma(\Delta), x \satisfies \Delta}, x \satisfies \chi : w \satisfies \phi }
\]
Without loss of generality, by Lemma \ref{lem:worldindependence} and by Lemma \ref{lem:packing}, each branch is then weakened and packed so that the reduction from the original BVS is as follows:
\[
\infer[\Uparrow]{ w \satisfies \Gamma(\Delta) : w \satisfies \phi }{x \satisfies \Delta : x \satisfies \chi & w \satisfies \Gamma(\Delta\fatsemi \chi) : w \satisfies \phi}
\]
These reductions are captured by $\system{VBI}$ as an instance of  $\cont$ followed by $\cut$.

\bigskip \noindent \textsc{Extract-closed Resolutions.}
Though stipulated that reductions begin by unpacking the context, the unpacking is left trivial (i.e., the empty unpacking) in the following case analysis when it is not required.

\begin{itemize}
    \item[$\land$ -] Reductions beginning with the $\land$-clause are as follows:
    \[
    \infer[\Uparrow]{w \satisfies \Gamma : w \satisfies \phi \land \psi }{
        \infer[\rrn \with]{
        w \satisfies \Gamma : (w \satisfies \phi) \with (w \satisfies \psi) 
        }{
        w \satisfies \Gamma : w \satisfies \phi & w \satisfies \Gamma : w \satisfies \psi
        }
    }
    \]
    These reductions are captured by $\system{VBI}$ as  $\rcl \land$.

    \item[$\lor$ -] Reduction beginning with the $\lor$-clause are as follows:
    \[
    \infer[\Uparrow]{w \satisfies \Gamma : w \satisfies \phi \lor \psi }{
        \infer[\rrn \with]{
        w \satisfies \Gamma : (w \satisfies \phi_1) \parr (w \satisfies \phi_2) 
        }{
        w \satisfies \Gamma : w \satisfies \phi_i
        }
    }
    \]
    These reductions are captured by $\system{VBI}$ as  $\rcl \lor$.

     \item[$\to$ -] Reductions beginning with the $\to$-clause are as follows:
     \[
     \infer[\Uparrow]{ w \satisfies \Gamma : w \satisfies \phi \to \psi}{ w \satisfies \Gamma : w \preceq u \Rightarrow (u \satisfies \phi \Rightarrow u \satisfies \psi) }
     \]
     By the invertability of $\rrn \Rightarrow$, this is continued to yield the following:
     \[
    \infer[\Uparrow]{w \satisfies \Gamma : w \satisfies \phi \to \psi}{ w \satisfies \Gamma, w \preceq u , u \satisfies \phi : u \satisfies \psi  }
     \]
    Without loss of generality, this reduction is continued by persistence. This follows by Lemma \ref{lem:worldindependence} as, if not, then $w \satisfies \Gamma$ and $w \prec u$ may be removed without loss of completeness, but this removal can still happen after persistence. Moreover, by Lemma \ref{lem:packing}, the reduction is thence continued by a packing. In total, the reduction is as follows:
         \[
    \infer[\Uparrow]{w \satisfies \Gamma : w \satisfies \phi \to \psi}{ w \satisfies \Gamma \fatsemi \phi : u \satisfies \psi  }
     \]
     These reductions are captured by $\system{VBI}$ as  $\rcl \to$.

   \item[$\top$ -] Reductions beginning with the $\top$-clause are as follows:
   \[
   \infer[\Uparrow]{w \satisfies \Gamma : w \satisfies \top }{w \satisfies \Gamma : \square}
   \]
   Without loss of generality, the proof-search terminates by $\rrn \square$-axiom. These reductions are captured in $\system{VBI}$ as $\rcl \top$.

         \item[$\bot$ -] Reductions beginning with the $\bot$-clause are as follows:
         \[
         \infer[\Uparrow]{ w \satisfies \Gamma: w \satisfies \bot}{w \satisfies \Gamma : w = \bot}
         \]
         Without loss of generality this is continued by the same reduction in reverse. But this is equivalent to doing no reduction at all. 
    
    \item[$*$ -] Reductions beginning with the $*$-clause are as follows:
    \[
    \infer[\Uparrow]{\Sigma_{w , \Gamma} : w \satisfies \phi*\psi}{
        \infer[\Uparrow]{\Sigma_{w , \Gamma} :  R(w,u,v) \with (u \satisfies \phi) \with (v \satisfies \psi)}{
            \Sigma_{w , \Gamma} :  R(w,u,v) 
            &
            \Sigma_{w , \Gamma} : u \satisfies \phi 
            &
            \Sigma_{w , \Gamma} :  v \satisfies \psi
        }
    }
    \]
    This can only lead to a proof if there were $R(w,u,v), u \satisfies \phi, v \satisfies \psi \in \Sigma_{w , \Gamma}$, in which case $\Gamma = \Gamma' \fatsemi (\Delta \fatcomma \Delta')$. But then, without loss of generality, $\ax$ is applied to one branch and Lemma \ref{lem:worldindependence} to the others, so that the reduction yields the following:
      \[
    \infer[\Uparrow]{\Sigma_{w , \Gamma' \fatsemi (\Delta_1 \fatcomma \Delta_2)} : w \satisfies \phi_1*\phi_2}{
       \Sigma_{u \satisfies \Delta_1} : u\satisfies \phi_1
       &
       \Sigma_{v \satisfies \Delta_2} : v\satisfies \phi_2
    }
    \]
    Without loss of generality, by Lemma \ref{lem:packing}, the reduction is continued by packing. These reductions are captured in $\system{VBI}$ as $\rcl \ast$.

    \item[$\wand$ -] Reductions beginning with the $\wand$-clause are as follows:
    \[
    \infer[\Uparrow]{ w \satisfies \Gamma : w \satisfies \phi \wand \psi}{
       w \satisfies \Gamma : R(w',w,u) \with u \satisfies \phi \Rightarrow w' \satisfies\psi        }
    \]
    By invertability of $\rrn \Rightarrow$ and $\lrn \with$, this is continued to yield the following:
    \[
     \infer[\Uparrow]{ w \satisfies \Gamma : w \satisfies \phi \wand \psi}{w \satisfies \Gamma,  R(w',w,u) , u \satisfies \phi : w' \satisfies\psi}
    \]
    Without loss of generality, by Lemma \ref{lem:packing}, this is continued with a packing. These reductions are captured in $\system{VBI}$ as $\rcl \wand$.
     \end{itemize}
     
    \textsc{Clauses Applied to the Context.} Each case begins with an unpacking that produces some assertion $x \satisfies \chi$ on which the clause defining the case is applied. 
    
    \begin{itemize}
    \item[$\land$ - ]  Reductions beginning with the $\land$-clause are as follows:
    \[
    \infer[\Uparrow]{\Pi_{w , \Gamma(\phi \land \psi)}, (x \satisfies \phi \land \psi) : w \satisfies \chi}{\Pi_{w , \Gamma(\phi \land \psi)} ,  (x \satisfies \phi), (x \satisfies \psi) : w \satisfies \chi}
    \]
    Without loss of generality, by Lemma \ref{lem:packing}, it is continued by a packing. These reductions are captured in $\system{VBI}$ as $\lcl \land$.

    \item[$\lor$ -]  Reductions beginning with the $\lor$-clause are as follows:
    \[
    \infer[\Uparrow]{\Pi_{\Gamma(\phi \lor \psi),x}, (x \satisfies \phi \lor \psi) : w \satisfies \chi}{
        \infer[\lrn \with]{
    \Pi_{\Gamma(\phi \lor \psi),x}, (x \satisfies \phi) \parr (x \satisfies \psi) : w \satisfies \chi
    }{
    \Pi_{\Gamma(\phi \lor \psi),x}, (x \satisfies \phi)  : w \satisfies \chi 
    &
    \Pi_{\Gamma(\phi \lor \psi),x}, (x \satisfies \phi) : w \satisfies \chi
    }
    }
    \]
   Without loss of generality, by Lemma \ref{lem:packing}, each branch is continued by a packing. These reductions are captured in $\system{VBI}$ as $\lcl \lor$.

    \item[$\to$ -] Reductions beginning with the $\to$-clause are as follows:
    \[
    \infer[\Uparrow]{\Pi_{w , \Gamma(\Delta \fatsemi \phi \to \psi), x} , (x \satisfies \Delta), (x \satisfies \phi \to \psi) : w \satisfies \chi}{\Pi_{w , \Gamma(\Delta \fatsemi \phi \to \psi), x} , (x \satisfies \Delta), \forall y (x \preceq y \Rightarrow (y \satisfies \phi \Rightarrow y \satisfies \psi) : w \satisfies \chi}
    \]
    The only choice of instantiation that can terminate in a proof is to instantiate the quantified world-variable as $x$. At this point the resulting sub-formula can be decomposed or else the resolution could be permuted with the next resolution. Hence, the reduction is continued as follows:
    \[
    \infer[\Uparrow]{\Pi_{w , \Gamma(\Delta \fatsemi \phi \to \psi), x} , (x \satisfies \Delta), (x \satisfies \phi \to \psi) : w \satisfies \chi}{
        \infer[\lrn\Rightarrow]{ \Pi_{w , \Gamma(\Delta \fatsemi \phi \to \psi), x} , (x \satisfies \Delta), (x \preceq x \Rightarrow (x \satisfies \phi \Rightarrow x \satisfies \psi) : w \satisfies \chi}{
            \Pi_{w , \Gamma(\Delta \fatsemi \phi \to \psi), x}, (x \satisfies \Delta) : x \satisfies \phi
            &
            \Pi_{w , \Gamma(\Delta \fatsemi \phi \to \psi), x}, (x \satisfies \Delta), (x \satisfies \psi) : w \satisfies \chi
        }
    }
    \]
     Without loss of generality, by Lemma \ref{lem:packing}, each branch is continued by a packing. These reductions are captured in $\system{VBI}$ as $\lcl \to$.

    \item[$\top$ - ] There are two possible reduction patterns beginning with the $\top$-clause. First, one may have the following:
    \[
    \infer[\Uparrow]{\Pi_{w , \Gamma(\Delta \fatsemi \aunit),x }, (x \satisfies \Delta), (x \satisfies \aunit) : w \satisfies \chi}{\Pi_{w , \Gamma(\Delta \fatsemi \aunit),x }, (x \satisfies \Delta) : w \satisfies \chi}
    \]
     Without loss of generality, by Lemma \ref{lem:packing}, each branch is continued by a packing. These reductions are captured in $\system{VBI}$ as $\weak$.
    Second, one may have the following:
    \[
    \infer[\Uparrow]{\Pi_{w , \Gamma(\Delta),x}, (x \satisfies \Delta) : w \satisfies \chi}{\Pi_{w , \Gamma(\Delta),x}, (x \satisfies \Delta), (x \satisfies \aunit) : w \satisfies \chi}
    \]
    Without loss of generality, by Lemma \ref{lem:packing}, each branch is continued by a packing. These reductions are captured in $\system{VBI}$ as $\lcl{\top}$.

    \item[$\bot$ -] Reductions beginning with the $\bot$-clause are as follows:
    \[
    \infer[\Uparrow]{\Pi_{w , \Gamma(\bot), x}, x \satisfies \bot : w \satisfies \chi}{\Pi_{w , \Gamma(\bot), x}, (x = \pi) : w \satisfies \chi}
    \]
    If another resolution is made then the the two resolution could have been permuted, unless the resolution was with the absurdity law, in which case the reduction continued to yield the following:
    \[
    \infer[\Uparrow]{\Pi_{w , \Gamma(\bot), x}, (x \satisfies \bot) : w \satisfies \chi}{\Pi_{w , \Gamma(\bot), x}, (x \satisfies \phi): w \satisfies \chi}
    \]
      Without loss of generality, by Lemma \ref{lem:packing}, each branch is continued by a packing. These reductions are captured in $\system{VBI}$ as $\lcl{\bot}$.

\item[$*$ -] There are two possible reduction patterns beginning with the $\top$-clause. First, one may have the following:
\[
\infer[\Uparrow]{\Pi_{w , \Gamma(\phi * \psi), x}, x  \satisfies \phi * \psi) : w \satisfies \chi }{\Pi_{w , \Gamma(\phi * \psi), x}, R(x,u,v), u \satisfies \phi, v \satisfies \psi : w \satisfies \chi}
\]
 Without loss of generality, by Lemma \ref{lem:packing}, each branch is continued by a packing. These reductions are captured in $\system{VBI}$ as $\lcl{\ast}^\mathsf{1}$.
Second, one may have the following:
    \[
    \infer[\Uparrow]{\Pi_{w , \Gamma(\Delta \fatcomma \I), x}, x  \satisfies \Delta * \I) , x \satisfies \Delta \fatcomma \I : w \satisfies \chi}{\Pi_{w , \Gamma(\Delta \fatcomma \I), x}, x  \satisfies \Delta * \I) , R(x,x,e), x \satisfies \Delta, e \satisfies \I : w \satisfies \chi}
    \]
   Without loss of generality, by Lemma \ref{lem:worldindependence} and Lemma \ref{lem:packing}, this is continued to yield the following:
    \[
    \infer[\Uparrow]{\Pi_{w , \Gamma(\Delta \fatcomma \I), x}, x  \satisfies \Delta * \I) , x \satisfies \Delta \fatcomma \I : w \satisfies \chi}{ w \satisfies \Gamma(\Delta) : w \satisfies \chi}
    \]
These reductions are captured in $\system{VBI}$ as $\lcl{\ast}^\mathsf{2}$.

\item[$\wand$ -] Reductions beginning with the $\wand$-clause are as follows:
\[
\infer[\Uparrow]{\Pi_{w , \Gamma(\Delta \fatcomma \Delta' \fatcomma \phi \wand \psi),x}, \Sigma, \Psi y \satisfies \Delta, u \satisfies \Delta', v \satisfies \phi \wand \psi : w \satisfies \chi}{ \Pi_{w , \Gamma(\Delta \fatcomma \Delta' \fatcomma \phi \wand \psi),x}, y \satisfies \Delta, u \satisfies \Delta', \Sigma : w \satisfies \chi}
\]
where $\Sigma := \{ R(x,y,z), R(z,u,v) \}$ and 
\[
\Psi := \forall a,b (R(b,v,a) \Rightarrow (a \satisfies \phi \Rightarrow b \satisfies \psi)) \}
\]
There is only one choice of instantiation for $a$ and $b$ that can terminate in a proof, which yields the the following reduction pattern:
\[
\infer[\Uparrow]{\Pi_{w , \Gamma(\Delta \fatcomma \Delta' \fatcomma \phi \wand \psi),x}, \Sigma, (y \satisfies \Delta), (u \satisfies \Delta'), (v \satisfies \phi \wand \psi) : w \satisfies \chi}{ \Pi_{w , \Gamma(\Delta \fatcomma \Delta' \fatcomma \phi \wand \psi),x}, \Sigma, (y \satisfies \Delta), (u \satisfies \Delta'), \Psi' : w \satisfies \chi }
\]
where 
\[
\Psi' :=  R(x,y,z), R(z,u,v), R(z,v,u) \Rightarrow (u \satisfies \phi \Rightarrow z \satisfies \psi)
\]
The sub-formula is immediately decomposed or else this resolution and the next could have been permuted. Hence, the reduction continues to yield sub-goals 
\[
\Pi_{w , \Gamma(\Delta \fatcomma \Delta' \fatcomma \phi \wand \psi),x}, R(x,y,z), R(z,u,v), (y \satisfies \Delta), (u \satisfies \Delta') : u \satisfies \phi
\]
and 
\[\Pi_{w , \Gamma(\Delta \fatcomma \Delta' \fatcomma \phi \wand \psi),x}, R(x,y,z), R(z,u,v), (y \satisfies \Delta), (u \satisfies \Delta') : w \satisfies \chi
\]
Without loss of generality, by Lemma \ref{lem:packing}, each branch is continued by a packing. These reductions are captured in $\system{VBI}$ as $\lcl{\wand}$.
\end{itemize}
    
    \textsc{Case Analysis on the Frame Laws.} The frame laws are unitality of $e$, commutative of $R$, associativity of $R$, persistence of $\prec$, dominance of $\prec$ and the absurdity of $\pi$. Except for the first three frame laws, the clauses can only be used after a particular resolution has occurred that introduces the appropriate atom, and these cases have been considered above; for exmple, persistence requires $w \prec u$ to appear in the context, which can only happen if $w \satisfies \phi \to \psi$ was resolved in the extract. We consider here the remaining cases. 
   
    \begin{itemize}
\item[Unit. - ] Reductions beginning with unitality are as follows:
\[
\infer[\Uparrow]{ \Sigma_{\Gamma(\Delta),x }, x \satisfies \Delta : w \satisfies \chi}{\Sigma_{\Gamma(\Delta),x } x \satisfies \Delta, R(x,x,e) : w \satisfies \chi}
\]
Without loss of generality, by Lemma \ref{lem:packing}, the reduction is continued with a packing. But, this simply yields the original sequent. Otherwise, it may be that a weakening  on $x \satisfies \Delta$ and $R(x,x,e)$  is performed and then the packing occurs. These reductions are captured in $\system{VBI}$ as $\lcl{\I}$.

\item[Comm. -]  Reductions beginning with commutativity of $R$ are as follows:
\[
\infer[\Uparrow]{\Pi_{\Gamma(\Delta \fatcomma \Delta'),x}, R(x,u,v), u \satisfies \Delta, v \satisfies \Delta' : w \satisfies \chi}{ \Pi_{\Gamma(\Delta \fatcomma \Delta'),x}, R(x,v,u), u \satisfies \Delta, v \satisfies \Delta' : w \satisfies \chi}
\]
Without loss of generality, by Lemma \ref{lem:packing}, this is continued by a packing. These reductions are captured in $\system{VBI}$ as $\lcl{\comm}$.

\item[Asso. - ] Reductions beginning with associativity of $R$ are as follows:
\[
\infer[\Uparrow]{\Pi_{\Gamma(\Delta\fatcomma(\Delta'\fatcomma \Delta''))}, R(x,y,z), y \satisfies \Delta, R(z,u,v), u \satisfies \Delta', v\satisfies \Delta'' : w \satisfies \chi }{\Pi_{\Gamma(\Delta\fatcomma(\Delta'\fatcomma \Delta''))}, R(x,a,v), y \satisfies \Delta, R(a,z,u), u \satisfies \Delta', v\satisfies \Delta'' : w \satisfies \chi}
\]
Without loss of generality, by Lemma \ref{lem:packing}, this is continued by a packing. These reductions are captured in $\system{VBI}$ as $\lcl{\asso}$.
    \end{itemize}

 \textsc{Case Analysis of the Structural Rules.}  There are instances of the structural rules that do not result in a change of sequent after packing; for example, permuting meta-atoms that are not assertions is without effect. In the following we restrict attention to the cases where the use of the structural rule affects the packing of the sequent.

\begin{itemize}
\item[\rn{exch} - ] Reductions beginning with an exchange are as follows:
\[
\infer[\Uparrow]{\Pi_{\Gamma(\Delta\fatsemi\Delta'), x}, x  \satisfies \Delta, x \satisfies \Delta': w \satisfies \chi}{\Pi_{\Gamma(\Delta\fatsemi\Delta'), x}, (x  \satisfies \Delta'), (x \satisfies \Delta'): w \satisfies \chi}
\]
Without loss of generality, by Lemma \ref{lem:packing}, this is continued by a packing. These reductions are captured in $\system{VBI}$ as $\exch_\mathsf{1}$.

\item[\rn{cont} - ] Reductions beginning with contractions are as follows:
\[
\infer[\Uparrow]{ \Pi_{\Gamma(\Delta), x}, x  \satisfies \Delta : w \satisfies \chi}{\Pi_{\Gamma(\Delta), x}, x \satisfies \Delta, x \satisfies \Delta : w \satisfies \chi}
\]
Without loss of generality, by Lemma \ref{lem:packing}, this is continued by a packing. These reductions are captured in $\system{VBI}$ as $\cont$.

\item[\rn{weak} -] Reductions beginning with weakening are as follows:
\[
\infer[\Uparrow]{\Pi_{\Gamma(\Delta \fatsemi \Delta'), x}, x  \satisfies \Delta x , x \satisfies \Delta': w \satisfies \chi}{\Pi_{\Gamma(\Delta), x}, x \satisfies \Delta : w \satisfies \chi}
\]
Without loss of generality, by Lemma \ref{lem:packing}, this is continued by a packing. These reductions are captured in $\system{VBI}$ as $\cont$.
\end{itemize}
This completes the proof.

\end{proof}

It useful to make precise how to read BI content from a BVS, which is understood as its \emph{state}.

\begin{Definition}[State]
    Let $\langle \Sigma_{\logic{BI}}, (w \satisfies \Gamma) : (w \satisfies \phi) \rangle$ be a BVS, its state is given by the following: $\Gamma : \phi $.
\end{Definition}

Hence each rule in \system{VBI} can be understood directly as a rule about sates. Since it is sound and complete for BVSs, it is then a calculus of entailment; for example, the $\rcl{\land}$-rule captures the following action on states:
\[
\infer{\Gamma \models \phi \land \psi}{\Gamma \models \phi & \Gamma \models \psi}
\]
We return to this observation in Section \ref{sec:snc}, below.

To complete this section, we return to an earlier claim made in Section \ref{subsec:encoding}: the $\I$-clause of satisfaction may be dropped without loss of generality. The class of frames encoded by $\Sigma_{\logic{BI}}$ is actually more general than the class of BI-frames, but the same BVSs hold and it suffices to demonstrate the approach to completeness. Let $\Phi_{\I} := \forall x(x\satisfies I \iff e \preceq x)$, we claim $\Sigma_{\logic{BI}}, (w \satisfies \Gamma) \proves (w  \satisfies \I)$ iff $\Sigma_{\logic{BI}}, \Phi_I, (w \satisfies \Gamma) \proves (w  \satisfies \I)$.  This follows from the fact that $\Sigma_{\logic{BI}}, \Phi_{\I}, (w \satisfies \Gamma) \vdash (w \satisfies \I)$ iff $\Gamma \proves \I$ which is what we would expect for a model of BI, in which case we already have $\Sigma_{\logic{BI}}, (w \satisfies \Gamma) \proves (w  \satisfies \I)$.  In short, the $\I$-clause can be removed from BI-frames without loss of generality when encoding in the meta-logic because the sequent calculus rule governing $\I$ requires that $\I$ is already part of the context --- indeed, this is the same reason satisfaction of atoms could be eliminated with impunity in Section \ref{subsec:encoding}. The other atomic rules, such as $\top$ and $\bot$ do not satisfy this condition, therefore their clauses are required.

\section{Soundness and Completeness} \label{sec:snc}

In Section \ref{sec:bi} we provided the sequent calculus \system{sLBI} for BI-provability; and, in Section \ref{sec:modeltheoryquaclassicaltheory}  we provided the sequent calculus $\system{VBI}$ for BI-validity. In Section \ref{sec:proofsearch}, we discussed the reductive reading of sequent calculi from which a notion of computation is inherited --- viz. proof-search. This notion of computation may be regarded as a \emph{transition system} on sequents. In this section we study the equivalence of the transition system for provability and the transition system for validity. 

There are many notions of equivalence between transition system. Here we are concerned with the subset that pertain to \emph{behavioural} equivalence; that is, how transitions in one system may be understood as transitions in the other. The finest notion of behavioural equivalence is \emph{bisimulation}.

\begin{Definition}[Bisimulation of Transition Systems]
   Let $\mathfrak{T}_1:=\langle \mathbb{S}_1, \rightsquigarrow_1 \rangle$ and  $\mathfrak{T}_n:=\langle \mathbb{S}_2, \rightsquigarrow_2 \rangle$ be transition systems. A relation $R \subseteq \mathbb{S}_1 \times \mathbb{S}_2$ is a bisimulation between $\mathfrak{T}_1$ and $\mathfrak{T}_2$ iff, for any $p \in \mathbb{S}_1$ and $q \in \mathbb{S}_2$ such that $pRq$,
   \begin{itemize}
       \item if there is $p' \in \mathbb{S}_1$ such that $p \rightsquigarrow_1 p'$, then there is $q \in \mathbb{S}_2$ such that $q \rightsquigarrow_2 q'$
       \item  if there is $q' \in \mathbb{S}_2$ such that $q \rightsquigarrow_2 q'$, then there is $p \in \mathbb{S}_1$ such that $p \rightsquigarrow_1 p'$.
   \end{itemize}
   The transition systems are bisimilar iff there is a bisimulation between them.
\end{Definition}

As in Section \ref{sec:proofsearch}, a proof system cannonically determines a transition system. 

\begin{Theorem}\label{thm:bisimilar}
    Provability is bisimilar to validity,
    \[
\Gamma \proves \phi \text{ bfs } \Gamma \models \phi
\]
\end{Theorem}
\begin{proof}
Let $\sim$ be the least relation satisfying the following:
\[
\{(\Gamma : \phi) \} \sim \{ \Sigma_{\logic{BI}}, w \satisfies \Gamma: w \satisfies \phi \}
\]
By observing the symmetry of the rules in Figure \ref{fig:slbi} and Figure \ref{fig:vbi}, we see that $\sim$ is a bisimulation.
\end{proof}

By unpacking the soundness proof of Theorem \ref{thm:vbi} within the proof of Theorem \ref{thm:bisimilar}, one recovers the usual inductive proof of soundness --- in the above style, the proof is a s\emph{simulation}. The contribution of this paper is to demonstrate an analogous technique for proving completeness. In this case, unpacking the completeness proof of Theorem \ref{thm:vbi} within the proof of Theorem \ref{thm:bisimilar} one recovers a \emph{co-inductive} proof of completeness. This highlights the duality between soundness and completeness. 

\begin{Corollary} \label{cor:snc}
Provability is extensionally equivalent to validity,
\[
\Gamma \proves \phi \qquad \text{ iff } \qquad \Gamma \models \phi
\]
\end{Corollary}
\begin{proof}
Follows immediately from Lemma \ref{lem:lbiprime} and Theorem \ref{thm:bisimilar}.
\end{proof}

\section{Relationship to Other Semantics}  \label{sec:previous}

The model theory of BI has been a subject of study for a while, and in this section we survey some earlier results such as the monoidal semantics dicussed by O'Hearn and Pym \cite{OHearn1999}, the Grothendieck topological semantics by Pym et al. \cite{Pym2004bi}, and the uniform approach of Docherty and Pym~\cite{Docherty2018a,DP2018,Docherty2019},  the relational semantics of Galmiche~et ~al.~\cite{Galmiche2005}. 

Throughout we use the notation of the meta-logic as \emph{bona fide} notation of the ambient logic; for example, we will use $\Rightarrow$ to denote contingency, $\with$ for conjunction, and $\parr$ for disjunction, without further reference.

\subsection{Preordered Commutative Monoids}

The relation $R$ in BI-frames can seem a little bit obscure, but a particularly simple way of defining it is through a monoidal product that takes the notion of decomposition of a state literally.

\begin{Definition}[Preordered Commutative Monoid]
 A PCM is a structure $\model = \langle \uni, \preceq, \circ, e \rangle$ in which $\preceq$ is a preorder, $\circ$ is a commutative monoidal product on $\uni$ with unit $e$;  that is, a preordered commutative monoid.
\end{Definition}

The algebraic reading of BI with ordered monoids is entirely coherent with Gabbay's theory of fibration, which determines the bifunctoriality condition:
\[
m \preceq m' \with  n \preceq n' \Rightarrow m \circ m' \preceq n \circ n'
\]

Ordered monoids are a particular case of BI-frames that arise by inheriting the preorder --- viz. defining $R(w,u,v) \iff w = u \circ v$.

\begin{Definition}[Monoid Model] \label{def:resourcemodel}
    A monoid model is a pair $\langle \model, \ll - \rr \rangle$ in which $\model := \langle \mathbb{M},  \preceq, \circ, e \rangle$ is a PCM and $\ll - \rr: \mathbb{M} \to \Atoms$ is an interpretation that is bifunctorial and atomically persistent.
\end{Definition}

Let $\mathcal{M}$ be the set of monoid algebras, then define a restriction of the semantics so far studied as follows:
\[
\Gamma \models[M] \phi \iff \forall \model[R] \in \mathcal{M} \, \forall u \in \uni \, (w \satisfies \Gamma \Rightarrow w \satisfies \phi)
\]

The soundness of monoid semantics (i.e., $\Gamma \proves \phi \Rightarrow \Gamma \models[M] \phi$) has been known for a while (see, for example, Pym \cite{Pym2002}) and is easy to prove using familiar methods, but completeness has remained open. Under the provided encoding, the monoidal semantics is contained in the semantics of this paper.

\subsection{The Consistency Semantics}

Traditionally only consistent formulas are taken to have meaning; that is, one usually considers a variant of satisfaction that proscribes the satisfaction of absurdity (i.e., $\bot$). Let $\satisfies^\top$ be the relation determined by the clauses in Figure \ref{fig:sat:BI} replacing the $\bot$-clause with the following:
\[
w \satisfies^\top \bot \quad \text{ iff } \quad  w \not \in \uni
\]

As before, the satisfaction relation determines a semantics:
\[
\Gamma \models^\top \phi \iff \forall \model[R] \in \mathcal{M} \, \forall u \in \uni \, (w \satisfies^\top \Gamma \Rightarrow w \satisfies^\top \phi)
\]

A positive result regarding the completeness with respect to this specialization was previously known:

\begin{Theorem}[Pym et al. \cite{Pym2004bi}, Pym \cite{Pym2002}]
If $\Gamma \models^\top \phi$ without $\bot$, then $\Gamma \proves \phi$ without $\bot$.
\end{Theorem}

The proof proceeds by the traditional method of a term-model construction. A stronger statement with respect to this satisfaction relation cannot be made:

\begin{Lemma}[Pym et al. \cite{Pym2004bi}, Pym \cite{Pym2002}] \label{lem:consistentmodelsincomplete}
 Let $\phi$ and $\psi$ be valid and be such that $\phi *\psi$ are valid, then define the following:
 \[
 \Gamma := (\phi \wand \bot) \to \bot \fatsemi (\psi \wand \bot) \to \bot \qquad \chi := ((\phi * \psi) \wand \bot) \to \bot
 \]
 For any instance it is the case that $\Gamma \models^\top \chi$, but not the case that $\Gamma \proves \chi$.
\end{Lemma}
\begin{proof}
One can check by proof-search that $\Gamma \proves \chi$ is not true, so it only remains to witness $\Gamma \models^\top \chi$. It is routine to verify that $x \satisfies^\top (\theta \wand \bot) \to \bot$ if and only if there is $y$ such that $y \satisfies^\top \theta$. Since $\phi$ and $\psi$ and $\phi*\psi$ are valid, any world suffices to witness that for an arbitrary $w$ it is the case that $w \satisfies \Gamma \with w \satisfies \chi$, which is stronger than $w \satisfies \Gamma  \Rightarrow w \satisfies \chi$.
\end{proof}
\begin{Lemma}
   Let $\phi := \top$ and $\psi := \top \wand \top$. Formulas $\phi$, $\psi$, and $\phi * \psi$ are all valid in BI.
\end{Lemma}

The form of Lemma \ref{lem:consistentmodelsincomplete} is pathological in that it expresses the incompatibility of the consistency condition with the totality of the monoids: if there are $u$ and $v$ such that $u \satisfies^\top A$ and $v \satisfies^\top A \wand \bot$, then $u \circ v \satisfies^\top A \fatcomma A \wand \bot$, but then $u \circ v \satisfies \bot$, which is absurd.

\subsection{The Inconsistency Semantics}

Since completeness fails for the consistency semantics, one can make a slight concession to the absurd: including a distinguished element $\pi$ dominating the algebra (i.e., $\forall w \in \mathbb{M} \Rightarrow w \preceq \pi$) that satisfies absurdity; this choice delivers the satisfaction relation given in Figure \ref{fig:sat:BI}. One may also substitute the equality for the preorder in the $\bot$-clause to form a candidate semantics. In either case, completeness fails:
\begin{Lemma}[Pym \cite{Pym2002}]
   Let $\phi = ((\psi \wand \bot) \wand \bot) \lor (\psi \wand \bot)$, then $e \satisfies^\bot \phi$ but $\phi$ is not valid in BI.
\end{Lemma}
\begin{proof}
The invalidity of $\psi$ can be proved by proof-search in the sequence calculus. By the definition of satisfaction, $e \satisfies^\bot \phi$ if and only if $e \satisfies^\bot \psi \wand \bot$ or $e \satisfies^\bot (\psi \wand \bot) \wand \bot$, we proceed by case analysis. First, if $e \satisfies \psi \wand \bot$, then the claim is trivially satisfied since it assumes one of the disjuncts. Second, if $e \not \satisfies \psi \wand \bot$, then the claim $e \satisfies (\psi \wand \bot) \wand \bot$ is equivalent to the following statement: for all $u$, if $u \satisfies^\bot \psi \wand \bot$, then $u = \pi$. This is, in turn, equivalent to the claim that  all $u$ there is a $v$ such that $v \satisfies^\bot \psi$ with $v \neq \pi$ or $u =\pi$. This is equivalent to the hypothesis as it may be unpacked to say: there is $u \in \mathbb{M}$ such that $u \satisfies^\bot \psi$ and $u \neq \bot$. 
\end{proof}

Consequently, one must modify the clause for disjunction too, effectively using Beth's clause instead of Kripke's. A term model construction exists with respect to topological monoids (see Pym \cite{Pym2002}), and more generally to the Grothendieck sheaf-theoretic models studied by Pym et al. \cite{Pym2004bi}. The position of Beth's clause with respect to the approach to completeness in this paper is discussed in Section~\ref{sec:beth}. 

\subsection{Partial and Non-deterministic Monoids} \label{subsec:simon} Another possibility for salvaging either the consistency or the inconsistency semantics is to make the monoidal product partial. Does this result in an adequate semantics with Kripke clause for disjunction? The question was answered positively by Docherty and Pym \cite{Docherty2018,Docherty2019a, Docherty2019} through an argument that makes use of a Stone-type duality. The authors simultaneously considered the option of having non-deterministic monoidal products, a consideration that arises naturally from the setting up of a uniform metatheory for bunched logics by extending the metatheory for intuitionistic layered graph logic \cite{Docherty2016,Docherty2018a}.

These partial and non-deterministic models are essentially the same as the models in this paper, but expressed as monoids rather than with a relation. Curiously, the motivation for the definition differs: in this paper, the models captured the minimal structure required for the simulation delivering completeness to take place, rather than from intuition about what a model of BI should look like. A stylistic consequence is that Definition \ref{def:bialg} requires persistence on formulas, whereas traditionally one would state atomic persistence along side other sufficient conditions that collectively deliver persistence. Given the motivation for the models in this paper, the sufficiency (and, possibly, necessity) of these conditions is a \emph{post-hoc} result about models, rather than the \emph{a priori} definition of them.

The structures involved in the semantics of Docherty and Pym \cite{Docherty2018,Docherty2019a, Docherty2019} are similar to the ordered monoids above except rather than have a unit $e$, they have a set of elements $E$ at least one of which is a unit, which further satisfies the following: 
\[
\underbrace{ e \in E \with e' \succeq e \Rightarrow e' \in E}_{\text{Closure}} \qquad \underbrace{ e \in E \with x \in y \circ e \Rightarrow y \preceq x}_{\text{Coherence}}
\]
\[
\underbrace{t' \succeq t \in x \circ y \with w \in t' \circ z \Rightarrow  \exists s,s',w'( s' \succeq s \in y \circ z \with  w \succeq w' \in x \circ s' )}_{\text{Strong Associativity}}
\]

Doherty and Pym begin by considering variations of the clauses for satisfaction; one replaces the clauses for $\I,*$ and $\wand$ with the following:
\[
\begin{array}{ccc}
x \satisfies \I & \text{ iff } & x \in E \\
x \satisfies \phi * \psi & \text{ iff } & \text{ there exists } x',y,z \text{ st. } x \succeq x' \in y \circ z, y\satisfies \phi \text{ and } z \satisfies \psi \\
    x \satisfies \phi \wand \psi & \text{ iff } & \text{ for any } x',y,z, \text{ if } x \preceq x', z \in x' \circ y \text{ and } y \satisfies \phi, \text{ then } z \satisfies \psi \\
\end{array}
\]

As above, given an interpretation, such structures have been shown to be sound and complete for BI when persistent; and, moreover, one has soundness and completeness for related logics upon suitable augmentation (e.g., replacing the preorder with equality one produces models for Boolean BI \cite{Pym2002}).

These variations can indeed be treated with the approach to completeness in this paper. The clauses used here are considered a simplification that arises when one expects models to act directly on the world being considered, yielding the following non-deterministic clauses for $*$ and $\wand$:
\[
\begin{array}{ccc}
x \satisfies \phi * \psi & \text{ iff } & \text{ there exists } y,z \text{ st. } x \in y \circ z, y\satisfies \phi \text{ and } z \satisfies \psi \\
    x \satisfies \phi \wand \psi & \text{ iff } & \text{ for any } y,z, \text{ if } z \in x \circ y \text{ and } y \satisfies \phi, \text{ then } z \satisfies \psi \\
\end{array}
\]

Soundness and completeness requires persistent models, but checking that a model satisfies this criterion or constructing one that does can be challenging. Fortunately, there are results in the literature that address this issue.

In the deterministic case the problem can be resolved by assuming bifunctoriality, but generalizing the property to non-deterministic case is a delicate matter. Cao et al. \cite{Cao2017} have considered the following conditions:
\[
\begin{array}{c}
 z \in x \circ y \with z \preceq z' \Rightarrow \exists x', y'( z' \in x' \circ y' \with x \preceq x' \with  y \preceq y') \\
  z \in x \circ y \with x' \preceq x \with y' \preceq y \Rightarrow \exists z'(z' \preceq z \with z' \in x' \circ y')
\end{array}
\]

Assuming these properties, called \emph{Upward Closed} and \emph{Downward Closed}, respectively, one recovers soundness with both the direct and indirect clauses for $*$ and $\wand$, respectively.  Moreover, Cao et al. \cite{Cao2017} showed that any structure satisfying either condition together with \emph{Simple Associativity} --- $t \in x \circ y  \with w \in t \circ z \Rightarrow \exists s(s \in y \circ z \with w \in x \circ s)$ --- can be conservatively transformed into sound models of BI satisfying all three. Docherty and Pym \cite{Docherty2019a, Docherty2019} has further shown that strong associativity for the non-deterministic models suffices for the same result without assuming the model to be either upward or downward closed.

\subsection{The Relational Semantics}
 Galmiche et al. \cite{Galmiche2005} attempted to put the partial semantics within a more general framework, delivering a relational semantics. The structures are similar to those of this paper, but necessarily include a distinguished element $\pi$ satisfying absurdity, satisfying the following:
\[
\underbrace{R\pi xy }_{\pi\text{-max}} \qquad \underbrace{Ryx\pi \Rightarrow \pi \preceq y}_{\pi\text{-abs}}
\]

Moreover, the preorder is defined in terms of the relation (i.e., $x \preceq y  \iff  Ryxe$) , and there are some additional conditions beyond commutativity and associativity:
\[
     \underbrace{R(z,x,y) \with x \preceq x'  \Rightarrow R(z,x',y)}_{\text{Compatibility}} \qquad  \underbrace{R(z,x,y) \with z \preceq z' \Rightarrow R(z',x,y)}_{\text{Transitivity}}
\]

The relational structures form models under an interpretation $\ll - \rr$ of the atoms when they are atomically persistent and, for any world $w$ and atom $A$, if $\pi \preceq w$, then $w \in \ll A \rr$. The resulting semantics was shown sound and complete via a term-model construction, and the models are subsumed by the class of BI-frames.

Unfortunately, the completeness of the version corresponding to the total monoid semantics failed as it passed through a tableaux system that has not been shown to be sound and complete --- see, for example, Galmiche et al. \cite{Galmiche2019}, where an attempt to bridge the sequent calculus and tableaux system is discussed. 

The relational semantics is a sub-class of the semantics of this paper. For example, it is consistent with Definition \ref{def:bialg} that there be a distinguished element $\pi$ satisfying $\pi$-max and $\pi$-abs. In the presence of bifunctoriality one can derive the slightly weaker condition than comptability:
\[
   R(z,x,y) \with x \preceq x' \Rightarrow \forall z'( R(z',x',y) \Rightarrow z \preceq z')
\]

And with associativity one can derive a weaker condition than transitivity:
\[
R(z,x,y) \with z \preceq z' \Rightarrow \exists x'( x \preceq x' \with R(z',x',y))
\]

From these observations it is clear that the distinguishing feature of relational semantics over the semantics in this paper is to assert certain equalities that are not required of BI, though they are often natural and useful.

\section{Beth's Disjunction.} \label{sec:beth} Before Kripke's landmark paper on the semantics of intuitionistic logic (IL) \cite{Kripke1965}, there was Beth's \cite{Beth1956}. The problem with Beth's semantics was that it included a complicated clause for disjunction, but it happens that the semantics is, in a sense explained by Kripke \cite{Kripke1965}, an \emph{unfolding} of Kripke's semantics. To explain briefly whence Beth's treatment of disjunction comes, how it relates to Kripke's and what they both have to do with the proof-search methodology of this paper, attention is now restricted to the additive fragment of BI (i.e., to IL.) which has provability relation $\proves[IL]$.  For a detailed account on the meta-theory of IL see Dummett \cite{Dummett2000}.

Kripke's semantics can immediately be stated given the study of BI so far conducted. Let $\satisfies[K]$ be the additive fragment of $\satisfies$, then Kripke models are defined as follows:

\begin{Definition}[Kripke Frame]
    A Kripke frame is a structure $\langle \uni, \preceq \rangle$ in which $\preceq$ is a preorder that is persistent; that is,
    \[
    w \preceq u \Rightarrow (w \satisfies[K] \phi \Rightarrow u \satisfies[K] \phi)
    \]
    The set of all Kripke frame is $\mathcal{K}$.
\end{Definition}

As before, this induces a semantics:
\[
\Gamma \models[K] \phi \iff \forall \model \in \mathcal{K}\, (w \satisfies[K] \Gamma \Rightarrow w \satisfies[K] \phi)
\]

\begin{Theorem}[Kripke \cite{Kripke1965}]
$\Gamma \proves[IL] \phi \iff \Gamma \models[K] \phi$
\end{Theorem}

Stating Beth's semantics requires more work, but the notion of a Beth structure can already be given:

\begin{Definition}[Directed Tree]
A directed tree is a directed graph whose underlying graph is connected and acyclic.
\end{Definition}

\begin{Definition}[Beth Structure]
    A structure $\langle \uni, \leq \rangle$ is a Beth structure if there is a directed tree $\langle \uni, \preceq \rangle$ such that $\leq$ is the transitive closure of $\preceq$.
\end{Definition}

The relationship between Kripke structures (i.e., preorders) and Beth structures is simply that when the former is a directed tree, its transitive closure is the latter.

The definition of Beth frame is cumbersome only in that it requires an specific condition on paths, called \emph{barring}. This notion is in fact quite natural when considering the constructivist account of intuitionistic logic.

\begin{Definition}[Path]
Let $\langle \uni, \preceq \rangle$  be directed tree. A path is a sequence $(x_i)_I \in \uni$ such that for every $i \in I$ it is the case that $x_i \preceq x_{i+1}$. A path $(x_i)_i$ is though a point $y$ when there is $i \in I$ such that $x_i =y$.
\end{Definition}

\begin{Definition}[Bar]
Let $B$ be some set of points in a Beth structure and let $x$ be a point. The set $B$ bars the point $x$, denoted $B \mid x$, when every path through $x$ intersects $B$. 
\end{Definition}

The Beth satisfaction relation $\satisfies[B]$ is as $\satisfies[K]$ but with the following clause for disjunction: 
    \[
    w \satisfies \phi \lor \psi \iff \exists \mathbb{U} \subseteq \uni \text{ st. } \mathbb{U} \mid  w \with \forall u \in \mathbb{U}(u \satisfies \phi \parr u \satisfies \psi)
    \]
    
\begin{Definition}[Beth frame]
A Beth frame is a Beth structure $\langle \uni, \leq \rangle$ that is persistent and persistent through barring; that is, for any $w, u \in \uni$ and $U \subseteq \uni$, the following hold:
    \[
    w \leq u \Rightarrow (w \satisfies[B] \phi \Rightarrow u \satisfies[B] \phi)
    \]
\[
(U |w \with \forall u \in U \, u \satisfies[B] \phi) \Rightarrow w \satisfies[B] \phi   
\]
The set of all Beth frames is $\mathcal{B}$.
\end{Definition}

A semantics is once more induced from the class of frames and satisfaction:
\[
\Gamma \satisfies[B] \phi \iff \forall \model \in \mathcal{B}\, (w \satisfies[B] \Gamma \Rightarrow w \satisfies[B] \phi)
\]

\begin{Theorem}[Beth \cite{Beth1956}]
$\Gamma \proves[IL] \phi \iff \Gamma \models[B] \phi$
\end{Theorem}

For both Beth and Kripke the intuition is that worlds represent states of information about constructions. In the case of Beth, one thinks of reasoning as being done in time, which is taken to be divided into successive intervals, say \emph{days}; thus, for example, the set $\{ w \mid w \leq u \}$   represents all the possible states to which one can eventually advance, and to one of which one shall advance. In this reading, barring says one can only have encountered a disjunction $\phi \lor \psi$ if at some preceding day one knew $\phi$ or one knew $\psi$. This is precisely the constructive reading of the connective.

How does all this relate to the proof-search methodology? Kripke's clause for disjunction is precisely algebraic analogue of the rule for disjunction:
\[
\infer{\Gamma \proves \phi_1 \lor \phi_2}{\Gamma \proves \phi_i} \qquad  
        \raisebox{1cm}{\xymatrix@R=0.05cm{
         & \bullet \ar[r] & (\Gamma, \psi) \\
        (\Gamma, \phi \lor \psi) \ar[ur] \ar[dr] & & \\
        & \bullet \ar[r] & (\Gamma, \phi)
        }
        }
         \qquad \infer{w \satisfies \Gamma \Rightarrow w \satisfies \phi \lor \psi}{ w \satisfies \Gamma \Rightarrow w \satisfies \phi}
\]

However, this proof-theoretic definition for disjunction is not necessarily the most natural one. Since $\leq$ represents $\proves[IL]$ (by persistence), the claim $\Gamma \proves[IL] \phi_1 \lor \phi_2$ says that if at some day one know $\Gamma$, then at a latter day  one knows $\phi_1 \lor \phi_2$. To witness this, there must a sequence of days in which one constructs from the information in $\Gamma$ either $\phi_1$ or $\phi_2$, thus the rules for disjunction takes the following form: 
\[
\infer{\Gamma \proves \phi_1 \lor \phi_2}{\Gamma \proves \Delta & \Delta \proves \phi_i}
\]
Thus, from the semantic perspective, $w \satisfies \phi_1 \lor \phi_2$ if and only if for some collection of $\Delta$ defining a set of worlds $\mathbb{U}$, it is either the case that that for an arbitrary element $u$ from the set either $u \satisfies \phi_1$ or $u \satisfies \psi_2$. This is the barring condition in Beth's clause for disjunction.

\section{Conclusion} \label{sec:conclusion}

This paper demonstrates the soundness and completeness of an algebraic semantics for the logic of Bunched Implications that has a minimal amount of pre-conditions; that is, the semantics generalizes and unifies previous results that have captured more specific classes of algebra or satisfaction relations: all that is required is a ternary relation that is commutative and associative and a binary preorder that is generally persistent (or, equivalently, atomically persistent and bifunctorial). One immediate task following this work is the provision of a tableaux system that may enable algorithmic reasoning to be preformed within BI, a logic quite important to computing and verification.

The \emph{deus ex machina} delivering the proof of soundness and (strong) completeness 
is a paradigm shift: rather than working from the traditional perspective of deductive logic, one works from the point of view of reductive logic. The r\^{o}le of reductions in the semantical analysis cannot be trivalized as it is precisely considering the co-recrusive construction of the proof-search space that delivers the proof. Moreover, what has not be discussed except in the most cursory manner is how the class of BI-algebras is determined in the first place, but this can be done in reductive logic while it seemingly cannot be done in deductive logic. It is this that justifies the phrase \emph{semantical analysis} in the title of the paper; the contained proof is the synthesis (i.e., the reversal of the analysis) of the work. Future work includes providing a rigorous general methodology.

The complication with deductive logic as the working paradigm when proving completeness is that it necessitates a bird's eye view of validity; that is, one must construct models fully, be it via a tableaux counter-model construction or a term model construction, both of which commit one to the completeness of a particular model, and then show at once that the sequent calculus could not allow any further construction. By comparison, reductive logic proceeds from the worm's eye view thereby allowing one to take full advantage of the local correctness property of sequent calculi (i.e., handle generic instances of rules), meaning that one has a modular approach to the analysis and hence the synthesis of the model theory.  For example, in the completeness proof above one could omit all of the multiplicative rules and one has proved the completeness of the additive fragment of BI (i.e., of intuitionistic logic) with respect to generally persistent preorders. 

The above practical considerations are important, but there is also a moral argument for reductive logic: it is the way in which one uses logic. For example, one regards $\phi \land \psi$ as meaning that both $\phi$ \emph{ and } $\psi$ hold, but simply phrasing it as such constitutes a reduction --- indeed, one now requires knowing what it means for the components themselves to hold! It is for this reason that though the main result of the paper is the soundness and completeness theorem for BI, the thesis is that reductive logic is both a natural and powerful perspective on logic that may yield further insight into meta-theory where traditional deductive approaches are either incapable or intractable.

\subsection*{Acknowledgments}
We a grateful to Timo Lang and Simon Docherty for their thorough and thoughtful comments 
on this work.

\newpage

\bibliographystyle{asl}
\bibliography{bib}

\end{document}

%% file: lbi.tex
\newcommand{\lbi}{
	\begin{figure} \centering
			\fbox{
				\begin{minipage}{0.95\linewidth}
					\begin{minipage}[t]{0.3\linewidth}
						\begin{scprooftree}{0.85}
							\AxiomC{$\Gamma(\Delta) \proves \chi$}
							\RightLabel{$\weak$}
							\UnaryInfC{$\Gamma(\Delta \fatsemi\Delta') \proves \chi$}
						\end{scprooftree}
					\end{minipage}
					\begin{minipage}[t]{0.3\linewidth}
						\begin{scprooftree}{0.85}
							\AxiomC{$\Gamma \proves \phi$}
							\RightLabel{$\exch_{(\Gamma \equiv \Gamma')}$}
							\UnaryInfC{$\Gamma' \proves \phi$}
						\end{scprooftree}
					\end{minipage}
					\begin{minipage}[t]{0.3\linewidth}
						\begin{scprooftree}{0.85}
							\AxiomC{$\Gamma(\Delta\fatsemi\Delta) \proves \phi$}
							\RightLabel{$\cont$}
							\UnaryInfC{$\Gamma(\Delta) \proves \phi$}
						\end{scprooftree}
					\end{minipage}	\\
					
					\begin{minipage}[t]{0.4\linewidth}
						\begin{scprooftree}{0.85}
							\AxiomC{$\Delta' \proves \phi$}
							\AxiomC{$\Gamma(\Delta\fatcomma\psi) \proves \chi$}
							\RightLabel{$\lrn \wand$}
							\BinaryInfC{$\Gamma(\Delta\fatcomma\Delta'\fatcomma\phi \wand \psi) \proves \chi$}
						\end{scprooftree}
					\end{minipage} 
					\begin{minipage}[t]{0.28\linewidth}
						\begin{scprooftree}{0.85}
							\AxiomC{$\Gamma\fatcomma\phi \proves \psi$}
							\RightLabel{$\rrn\wand$}
							\UnaryInfC{$\Gamma \proves \phi \wand \psi$}
						\end{scprooftree}
					\end{minipage} 
				\begin{minipage}[t]{0.28\linewidth}
					\begin{scprooftree}{0.85}
						\AxiomC{$\Gamma(\munit)\proves \chi$}
						\RightLabel{$\lrn \I$}
						\UnaryInfC{$\Gamma(\top^*) \proves \chi$}
					\end{scprooftree}
				\end{minipage} 
				\\ 
					
					\begin{minipage}[t]{0.4\linewidth}
						\begin{scprooftree}{0.85}
							\AxiomC{$\Delta \proves \phi$}
							\AxiomC{$\Gamma(\Delta' \fatsemi\psi) \proves \chi$}
							\RightLabel{$\lrn \to$}
							\BinaryInfC{$\Gamma(\Delta \fatsemi \Delta' \fatsemi \phi \to \psi) \proves \chi$}
						\end{scprooftree}
					\end{minipage}
					\begin{minipage}[t]{0.28\linewidth}
						\begin{scprooftree}{0.85}
							\AxiomC{$\Gamma \fatsemi\phi \proves \psi$}
							\RightLabel{$\rrn \to$}
							\UnaryInfC{$\Gamma \proves \phi \to \psi$}
						\end{scprooftree}
					\end{minipage} 
					\begin{minipage}[t]{0.28\linewidth}
					\begin{scprooftree}{0.85}
						\AxiomC{$\Gamma(\aunit)\proves \chi$}
						\RightLabel{$\lrn \top$}
						\UnaryInfC{$\Gamma(\top) \proves \chi$}
					\end{scprooftree}
				\end{minipage}
				\\ 
			
				\begin{minipage}[t]{0.25\linewidth}
						\begin{scprooftree}{0.85}
							\AxiomC{$\Gamma(\phi\fatcomma\psi) \proves \chi$}
							\RightLabel{$\lrn *$}
							\UnaryInfC{$\Gamma(\phi * \psi) \proves \chi$}
						\end{scprooftree}
					\end{minipage}
					\begin{minipage}[t]{0.35\linewidth}
						\begin{scprooftree}{0.85}
							\AxiomC{$\Gamma \proves \phi$}
							\AxiomC{$\Gamma' \proves \psi$}
							\RightLabel{$\rrn *$}
							\BinaryInfC{$\Gamma\fatcomma\Gamma' \proves \phi*\psi$}
						\end{scprooftree}
					\end{minipage} 
					\begin{minipage}[t]{0.35\linewidth}
					\begin{scprooftree}{0.85}
						\AxiomC{$\Gamma ( \phi  ) \proves \chi$}
						\AxiomC{$\Gamma ( \psi ) \proves \chi $}
						\RightLabel{$\lrn \lor$}
						\BinaryInfC{$\Gamma(\phi \lor \psi) \proves \chi$}
					\end{scprooftree}
				\end{minipage}
				\\ 
				
					\begin{minipage}[t]{0.25\linewidth}
						\begin{scprooftree}{0.85}
							\AxiomC{$\Gamma(\phi\fatsemi\psi) \proves \chi$}
							\RightLabel{$\lrn \land$}
							\UnaryInfC{$\Gamma(\phi \land \psi) \proves \chi$}
						\end{scprooftree}
					\end{minipage} 
					\begin{minipage}[t]{0.4\linewidth}
						\begin{scprooftree}{0.85}
							\AxiomC{$\Gamma \proves \phi$}
							\AxiomC{$\Gamma \proves \psi$}
							\RightLabel{$\rrn \land$}
							\BinaryInfC{$\Gamma \proves \phi\land\psi$}
						\end{scprooftree}
					\end{minipage} 
							\begin{minipage}[t]{0.3\linewidth}
					\begin{scprooftree}{0.85}
						\AxiomC{$\Gamma \proves  \phi_i$}
						\RightLabel{$\rrn \lor$}
						\UnaryInfC{$\Gamma \proves \phi_1 \lor \phi_2$}
					\end{scprooftree}
				\end{minipage}
			\\ 
			
					\begin{minipage}[t]{0.24\linewidth}
						\begin{scprooftree}{0.85}
							\AxiomC{$\square$}
							\RightLabel{$\bot$}
							\UnaryInfC{$\Gamma(\bot) \proves \phi$ }
						\end{scprooftree}
					\end{minipage}
					\begin{minipage}[t]{0.24\linewidth}
						\begin{scprooftree}{0.85}
							\AxiomC{$\square$}
							\RightLabel{\ax}
							\UnaryInfC{$\rm p \proves \rm p$ }
						\end{scprooftree}
					\end{minipage}
					\begin{minipage}[t]{0.24\linewidth}
						\begin{scprooftree}{0.85}
							\AxiomC{$\square$}
							\RightLabel{$\rrn \I$}
							\UnaryInfC{$\munit \proves \I$ }
						\end{scprooftree}
					\end{minipage}
					\begin{minipage}[t]{0.24\linewidth}
						\begin{scprooftree}{0.85}
							\AxiomC{$\square$}
							\RightLabel{$\rrn \top$}
							\UnaryInfC{$\aunit \proves \top$ }
						\end{scprooftree}
					\end{minipage}
				\end{minipage}
				}
			\caption{System \system{LBI} \label{fig:lbi}}
		\end{figure}
		
		}

		\newcommand{\exchrules}{
		\begin{figure}[t]
		\fbox{
		\begin{minipage}{0.95\linewidth}
		
        \begin{minipage}{0.24\linewidth}
        \begin{scprooftree}{0.85}
        \AxiomC{$\Gamma(\Delta\fatsemi \aunit) \proves \chi$}
        \RightLabel{$\contadd$}
        \UnaryInfC{$\Gamma(\Delta) \proves \chi$}
        \end{scprooftree}
        \end{minipage}
        \begin{minipage}{0.24\linewidth}
        \begin{scprooftree}{0.85}
        \AxiomC{$\Gamma(\Delta\fatcomma \munit) \proves \chi$}
        \RightLabel{$\contmult$}
        \UnaryInfC{$\Gamma(\Delta) \proves \chi$}
        \end{scprooftree}
        \end{minipage}
        \begin{minipage}{0.24\linewidth}
        \begin{scprooftree}{0.85}
        \AxiomC{$\Gamma(\Delta) \proves \chi$}
        \RightLabel{$\weakadd$}
        \UnaryInfC{$\Gamma(\Delta\fatsemi \aunit) \proves \chi$}
        \end{scprooftree}
        \end{minipage}
        \begin{minipage}{0.24\linewidth}
        \begin{scprooftree}{0.85}
        \AxiomC{$\Gamma(\Delta) \proves \chi$}
        \RightLabel{$\weakmult$}
        \UnaryInfC{$\Gamma(\Delta\fatcomma \munit) \proves \chi$}
        \end{scprooftree}
        \end{minipage}
	
	   \begin{minipage}{0.47\linewidth}
        \begin{scprooftree}{0.85}
        \AxiomC{$\Gamma(\Delta' \fatcomma \Delta) \proves \chi$}
        \RightLabel{$\mcomm$}
        \UnaryInfC{$\Gamma(\Delta \fatcomma \Delta') \proves \chi$}
        \end{scprooftree}
        \end{minipage}
        	   \begin{minipage}{0.47\linewidth}
        \begin{scprooftree}{0.85}
        \AxiomC{$\Gamma(\Delta' \fatsemi \Delta) \proves \chi$}
        \RightLabel{$\acomm$}
        \UnaryInfC{$\Gamma(\Delta \fatsemi \Delta') \proves \chi$}
        \end{scprooftree}
        \end{minipage}

        	   \begin{minipage}{0.47\linewidth}
        \begin{scprooftree}{0.85}
        \AxiomC{$\Gamma((\Delta \fatcomma \Delta') \fatcomma \Delta'') \proves\chi$}
        \RightLabel{$\masso$}
        \UnaryInfC{$\Gamma(\Delta \fatcomma (\Delta' \fatcomma \Delta'')) \proves \chi$}
        \end{scprooftree}
        \end{minipage}
                \begin{minipage}{0.47\linewidth}
        \begin{scprooftree}{0.85}
        \AxiomC{$\Gamma((\Delta \fatsemi \Delta') \fatsemi \Delta'') \proves\chi$}
        \RightLabel{$\aasso$}
        \UnaryInfC{$\Gamma(\Delta \fatsemi (\Delta' \fatsemi \Delta'')) \proves \chi$}
        \end{scprooftree}
        \end{minipage}

		\end{minipage}
		}
		\caption{The Exchange Rules} \label{fig:exch}
		\end{figure}
		}